\documentclass{article}

\usepackage[UKenglish]{babel}
\usepackage{fontenc}
\usepackage[utf8]{inputenc}
\usepackage{hyperref}
\usepackage{authblk}
\usepackage[a4paper, total={6in, 8in}]{geometry}

\usepackage{enumerate}
\usepackage{graphicx}
\graphicspath{figures}
\usepackage{tikzfig}
\usepackage{tikz-cd}

\tikzstyle{gbox}=[box, draw=black, shape=rectangle, fill={zx_green}, tikzit fill={rgb,255: red,181; green,215; blue,181}]
\tikzstyle{hbox}=[box, draw=black, shape=rectangle, fill=yellow, minimum size=.55em]
\tikzstyle{box}=[draw=black, shape=rectangle, fill=white, minimum size=.95em, inner sep=0.15em, scale=0.85, font={\footnotesize}]
\tikzstyle{gn}=[draw=black, shape=circle, fill={zx_green}, draw=black, inner sep=0.7mm, minimum width=0pt, minimum height=0pt, tikzit fill={rgb,255: red,181; green,215; blue,181}]
\tikzstyle{rn}=[gn, fill={zx_red}, draw=black, tikzit fill={rgb,255: red,215; green,96; blue,96}]
\tikzstyle{gn_phase}=[shape=rectangle, fill={zx_green}, draw=black, minimum size=1.2em, rounded corners=0.5em, inner sep=0.2em, outer sep=-0.2em, scale=0.8, font={\footnotesize}, tikzit shape=circle, tikzit fill={rgb,255: red,181; green,215; blue,181}]
\tikzstyle{rn_phase}=[{gn_phase}, fill={zx_red}, draw=black, tikzit fill={rgb,255: red,215; green,96; blue,96}]
\tikzstyle{ltriang}=[rtriang, shape=isosceles triangle, fill=yellow, draw=black, shape border rotate=180]
\tikzstyle{rtriang}=[shape=isosceles triangle, fill=yellow, draw=black, isosceles triangle stretches=true, inner sep=0.8pt, minimum width=0.25cm, minimum height=2mm]
\tikzstyle{dtriang}=[rtriang, shape=isosceles triangle, fill=yellow, draw=black, shape border rotate=-90]
\tikzstyle{utriang}=[rtriang, shape=isosceles triangle, fill=yellow, draw=black, shape border rotate=90]
\tikzstyle{lmat}=[shape=signal, signal to=west, signal from=east, fill={zx_grey}, draw=black, minimum height=6pt, inner sep=1pt, font={\scriptsize  \boldmath}, tikzit fill=gray, tikzit category=GLA]
\tikzstyle{rmat}=[shape=signal, signal to=east, signal from=west, fill={zx_grey}, draw=black, minimum height=6pt, inner sep=1pt, font={\scriptsize  \boldmath}, tikzit fill=gray, tikzit category=GLA]
\tikzstyle{dmat}=[shape=signal, signal to=west, signal from=east, fill={zx_grey}, draw=black, minimum height=6pt, inner sep=1pt, font={\scriptsize  \boldmath}, tikzit fill=gray, tikzit category=GLA, rotate=90]
\tikzstyle{umat}=[shape=signal, signal to=east, signal from=west, fill={zx_grey}, draw=black, minimum height=6pt, inner sep=1pt, font={\scriptsize  \boldmath}, tikzit fill=gray, tikzit category=GLA, rotate=90]
\tikzstyle{lw}=[shape=isosceles triangle, isosceles triangle stretches=true, fill=black, draw=black, minimum width=0.4cm, minimum height=3mm, inner sep=1pt, shape border rotate=180]
\tikzstyle{rw}=[shape=isosceles triangle, isosceles triangle stretches=true, fill=black, draw=black, minimum width=0.4cm, minimum height=3mm, inner sep=1pt]
\tikzstyle{dw}=[shape=isosceles triangle, isosceles triangle stretches=true, fill=black, draw=black, minimum width=0.4cm, minimum height=3mm, inner sep=1pt, shape border rotate=-90]
\tikzstyle{uw}=[shape=isosceles triangle, isosceles triangle stretches=true, fill=black, draw=black, minimum width=0.4cm, minimum height=3mm, inner sep=1pt, shape border rotate=90]
\tikzstyle{lsplit}=[shape=isosceles triangle, isosceles triangle stretches=true, fill=white, draw=black, minimum width=0.4cm, minimum height=3mm, inner sep=1pt, shape border rotate=180]
\tikzstyle{rmerge}=[shape=isosceles triangle, isosceles triangle stretches=true, fill=white, draw=black, minimum width=0.4cm, minimum height=3mm, inner sep=1pt]
\tikzstyle{vsplit}=[shape=isosceles triangle, isosceles triangle stretches=true, fill=white, draw=black, minimum width=0.4cm, minimum height=3mm, inner sep=1pt, shape border rotate=90]
\tikzstyle{phase}=[draw=black, shape=rectangle, fill=white, minimum size=.95em, inner sep=0.15em, scale=0.85, font={\footnotesize}]
\tikzstyle{black node}=[draw=black, shape=circle, scale=0.3, fill=black, font={\footnotesize}]
\tikzstyle{d_split}=[shape=trapezium, fill=white, draw=black, minimum width=11pt, inner sep=1pt]
\tikzstyle{d_merge}=[shape=trapezium, fill=white, draw=black, minimum width=11pt, inner sep=1pt, rotate=180]
\tikzstyle{bluedot}=[draw, shape=circle, fill={zx_green}, scale=0.35, tikzit fill=blue]
\tikzstyle{wire label}=[font={\scriptsize}, auto]
\tikzstyle{new style 0}=[fill=white, draw=black, shape=circle]

\tikzstyle{braceedge}=[-, decorate, decoration={brace, amplitude=2mm, raise=-1mm}]
\tikzstyle{blue-edge}=[-, very thick, draw=bluegray]
\tikzstyle{hadamard edge}=[-, dashed, draw=blue]
\tikzstyle{thick-arrow}=[->, very thick]
\tikzstyle{thick-line}=[-, very thick]
\tikzstyle{red-edge}=[-, draw=fireenginered, very thick]
\tikzstyle{red-line}=[-, draw=fireenginered]
\tikzstyle{arrow}=[->]
\tikzstyle{c-arrow}=[->, dashed]
\tikzstyle{board}=[-, thick]

\definecolor{zx_grey}{RGB}{211,211,211}
\definecolor{zx_red}{RGB}{232,165,165}
\definecolor{zx_green}{RGB}{216,248,216}
\definecolor{bluegray}{rgb}{0.4, 0.6, 0.8}
\definecolor{fireenginered}{rgb}{0.81, 0.09, 0.13}

\usepackage{import}
\usepackage{subfiles}

\usepackage{mathtools}
\usepackage{physics}
\usepackage{amsfonts}
\usepackage{amssymb}

\usepackage{amsthm}

\newtheorem{Th}{Proposition}
\newtheorem{theorem}[Th]{Proposition}
\newtheorem{proposition}[Th]{Proposition} 
\newtheorem*{proposition*}{Proposition}

\begin{document}
\title{Differentiation of Linear Optical Circuits}
\author[1]{Giovanni de Felice}
\author[2]{Christopher Corlett}
\affil[1]{Quantinuum, 17 Beaumont street, OX1 2NA Oxford, UK}
\affil[2]{University of Bristol, NSQI, Tyndall Ave, BS8 1UD Bristol, UK}
\date{}
\maketitle

\begin{abstract}
    Linear optical circuits with single-photon sources offer a promising platform for quantum chemistry and machine learning.
    However, current applications are all based on support vector machines or gradient-free optimization methods.
    This paper develops classical and quantum algorithms for evaluating the analytic gradients of linear optical circuits 
    with respect to their phase parameters.
    First, we set up a general framework by characterising the class of observables whose expectation values can be estimated efficiently
    by sampling from a passive linear optical circuit with finitely many photons.
    We then show how to compute the gradients of the expectation values of a special class of ``non-interacting'' observables arising in full-counting-statistics.
    Our differentiation algorithm uses the Halmos dilation and requires evaluating two circuits of twice the size, using one additional photon.
    Building on the methods of full-counting-statistics, we show how to recover the gradients of arbitrary observables from the gradient of a non-interacting characteristic function.
    Throughout the paper, we compare the performance of classical and quantum algorithms on the same estimation problems,
    analysing the sampling complexity of the algorithms and suggesting different cases for which quantum speed-ups could be obtained.
\end{abstract}

\section{Introduction}%

Linear optical circuits are fundamental components in photonic computing, both classical \cite{shen_deep_2017, feldmann_parallel_2021}
and quantum \cite{knill_scheme_2001, aaronson_computational_2011}.
Equipped with single-photon sources and detectors, they produce distributions that are hard to reproduce on a classical computer \cite{aaronson_computational_2011}.
In a typical experiment \cite{maring_general-purpose_2023, pentangelo_high-fidelity_2023}, 
a state with finitely many photons is prepared, it evolves through a passive linear optical circuit, 
and we sample the photon-number distribution at the output.
The quantum features of these distributions have motivated several recent experiments using linear optics to perform variational quantum eigensolvers
\cite{peruzzo_variational_2014}, chemistry simulations \cite{sparrow_simulating_2018} and other machine learning tasks \cite{bartkiewicz_experimental_2020,saggio_experimental_2021, bradler_certain_2021, gan_fock_2022}.
However, all applications so far suffer from the inability to compute the gradients of linear optical circuits.

Variational methods for quantum optimization \cite{cerezo_variational_2021} 
have unified several approaches in the fields of quantum chemistry \cite{grimsley_adaptive_2019}, 
combinatorial optimization \cite{amaro_filtering_2022} and machine learning \cite{benedetti_parameterized_2019}.
In this framework, the objective function that needs to be optimised is defined in terms 
of the expectation values of observables over a parametrised quantum circuit \cite{cerezo_variational_2021}.
It is then useful to compute the gradients of these expectation values, 
as they allow us to update the parameters and make the model converge in practical tasks.
Since these gradients are in general hard to estimate classically, we need a method for computing 
them --- in-situ --- by sampling from a quantum computer.
The problem of computing gradients on quantum hardware has received considerable attention in the qubit setting \cite{mitarai_quantum_2018, schuld_evaluating_2019}.
We consider the same problem for discrete-variable photonics: 
linear optical circuits equipped with number-resolving detectors and arbitrary initial states of finitely many photons.

While the theory of observables has been developed extensively in the qubit setting \cite{wocjan_several_2006, bravyi_classical_2021, hamamura_efficient_2020}, 
it has not received much attention in the context of linear optics.
For example, we know that, when the input state is a bosonic product state, the probability of an individual event can be estimated efficiently on a classical computer \cite{aaronson_generalizing_2012}.
However, in the same setting, the complexity of estimating moments and higher-order observables of the distribution are unknown.
Here we propose a general notion of discrete observable whose expectation values can be estimated by sampling from a linear optical circuit.
This includes observables arising in \emph{full-counting-statistics} (FCS), a powerful method to extract the statistical information 
of particle-number distributions \cite{levitov_charge_1993, del_campo_universal_2018, ivanov_complexity_2020, fan_full-counting_2024}.
In FCS, we have access to the expectation values a class of observables parametrised by eigenvalues of the form $e^{i \lambda}$. 
The probabilities, moments and cumulants of the distribution are then recovered by integration and differentiation with respect to $\lambda$.
In the context of variational algorithms, we are not only interested in parametrizing the observable but also in optimising circuit parameters.

The basic parametrizable component in a linear optical circuit is a phase shift of angle $\theta$ which acts on the 
Fock space as the unitary operator $e^{i\hat{n} \theta}$; where $\hat{n}$ is the number operator.
Its derivative with respect to $\theta$ is easily seen to be $i\hat{n} e^{i\hat{n}\theta}$.
This poses difficulty because the number operator $\hat{n}$ is not unitary and not even a bounded operator. 
As a result, the derivative of a phase is not itself a physically realisable circuit.
Our approach to solving this problem is based on dilation theory \cite{halmos1950, shalit_dilation_2020}, 
which provides general methods for simulating quantum mechanical operators by unitary circuits on a bigger space.

This paper develops a framework for differentiating linear optical circuits and their observables.
We start in Section~\ref{sec:linear-optics} by introducing QPath, a graphical language for linear optics that captures both circuits and multi-particle operators such as the ladder and number operators \cite{de_felice_quantum_2023}.
This language reflects how the properties of linear optical circuits are captured in their representation as ``non-interacting'' operators, 
consisting of matrices of size linear in the number of modes $m$.
The dilation theory of finite-dimensional matrices lifts to the multi-particle setting and 
we show that one may recover a unitary operator from any multi-particle diagram.
In Section~\ref{sec:observables}, we define the physical quantities that we wish to differentiate: 
expectation values of discrete observables on the Fock space.
We give an explicit proof that these can be estimated by sampling from a linear optical circuit
to within an additive error $\pm \epsilon s_{n, \lambda}$ where $s_{n, \lambda}$ is the largest singular value of the observable 
on the truncated space with at most $n$ photons.
We then consider the consequences of this result for several natural estimation problems arising in linear optics,
and we compare the scaling with known, exact and approximate, classical algorithms.
In Section~\ref{sec:circuit-gradient}, we consider the problem of evaluating the gradient of the expectation values of non-interacting observables with respect to circuit parameters.
We show that these can be estimated by sampling from a second circuit of twice the size, using one additional photon.
We are however only able to bound the approximation error by $s_W^n$ where $s_W = \frac{1+\sqrt{5}}{2}$, the golden ratio, is the singular value 
of the underlying matrix of the number operator and we discuss how our result could be improved.
In Section~\ref{sec:observable-gradient} we show that the framework of the full-counting-statistics carries over to the variational setting 
and we are able to estimate the gradients of moments and higher-order observables from the gradient of the characteristic function.
We end with an example in Section~\ref{sec:example}, applying our methods to the full-counting-statistics of a tunable beam splitter.


\section{Linear optical diagrams}\label{sec:linear-optics}%

\subsection{Matrices}
Consider the class of processes generated by the \emph{phase shift} with parameter 
$\theta \in [0, 2\pi)$ and the \emph{beam splitter}. These are defined by a corresponding matrix.
\begin{align}\label{eq-phase}
    \begin{split}
        \scalebox{0.8}{\begin{tikzpicture}
	\begin{pgfonlayer}{nodelayer}
		\node [style=none] (6) at (2.5, 0.75) {};
		\node [style=none] (7) at (4.5, 0.75) {};
		\node [style=none] (8) at (2.5, -0.75) {};
		\node [style=none] (9) at (4.5, -0.75) {};
		\node [style=none] (10) at (3.5, 0) {$\left(e^{i \theta}\right)$};
		\node [style=none] (15) at (-4.5, 0) {};
		\node [style=none] (16) at (-1.5, 0) {};
		\node [style=none] (17) at (0, 0) {$=$};
		\node [style=none] (18) at (1.5, 0) {};
		\node [style=none] (19) at (5.5, 0) {};
		\node [style=none] (20) at (2.5, 0) {};
		\node [style=none] (21) at (4.5, 0) {};
		\node [style=new style 0] (22) at (-3, 0) {$\theta$};
		\node [style=none] (23) at (1.75, -2) {};
		\node [style=none] (24) at (6.25, -2) {};
		\node [style=none] (25) at (1.75, -5) {};
		\node [style=none] (26) at (6.25, -5) {};
		\node [style=none] (27) at (4, -3.5) {$\frac{1}{\sqrt{2}}\begin{pmatrix}i & 1 \\ 1 & i \end{pmatrix}$};
		\node [style=none] (28) at (-5, -2.5) {};
		\node [style=none] (29) at (-1, -2.5) {};
		\node [style=none] (30) at (0, -3.5) {$=$};
		\node [style=none] (31) at (1, -2.5) {};
		\node [style=none] (32) at (7, -2.5) {};
		\node [style=none] (33) at (1.75, -2.5) {};
		\node [style=none] (34) at (6.25, -2.5) {};
		\node [style=none] (36) at (1, -4.5) {};
		\node [style=none] (37) at (1.75, -4.5) {};
		\node [style=none] (38) at (7, -4.5) {};
		\node [style=none] (39) at (6.25, -4.5) {};
		\node [style=none] (40) at (-5, -4.75) {};
		\node [style=none] (41) at (-1, -4.75) {};
		\node [style=none] (42) at (-4, -3.5) {};
		\node [style=none] (43) at (-2, -3.5) {};
		\node [style=none] (44) at (-4, -3.75) {};
		\node [style=none] (45) at (-2, -3.75) {};
		\node [style=none] (46) at (-3, -3.5) {};
		\node [style=none] (47) at (-3, -3.75) {};
	\end{pgfonlayer}
	\begin{pgfonlayer}{edgelayer}
		\draw (6.center) to (8.center);
		\draw (8.center) to (9.center);
		\draw (9.center) to (7.center);
		\draw (7.center) to (6.center);
		\draw (18.center) to (20.center);
		\draw (21.center) to (19.center);
		\draw (15.center) to (22);
		\draw (22) to (16.center);
		\draw (23.center) to (25.center);
		\draw (25.center) to (26.center);
		\draw (26.center) to (24.center);
		\draw (24.center) to (23.center);
		\draw (31.center) to (33.center);
		\draw (34.center) to (32.center);
		\draw (36.center) to (37.center);
		\draw (39.center) to (38.center);
		\draw (42.center) to (44.center);
		\draw (44.center) to (45.center);
		\draw (45.center) to (43.center);
		\draw (43.center) to (42.center);
		\draw [in=120, out=0] (28.center) to (46.center);
		\draw [in=-180, out=45] (46.center) to (29.center);
		\draw [in=180, out=-45] (47.center) to (41.center);
		\draw [in=0, out=-135] (47.center) to (40.center);
	\end{pgfonlayer}
\end{tikzpicture}}
    \end{split}
\end{align}
Matrices can be composed in sequence, by matrix mutiplication, and in parallel using the direct sum $\oplus$. 
We depict these compositions for two matrices, $A$ and $B$, diagrammatically as
$$\scalebox{0.95}{\begin{tikzpicture}
	\begin{pgfonlayer}{nodelayer}
		\node [style=none] (0) at (-15, 0) {};
		\node [style=none] (2) at (-14, 0) {};
		\node [style=none] (4) at (-12, 0) {};
		\node [style=none] (20) at (-14, 1) {};
		\node [style=none] (21) at (-12, 1) {};
		\node [style=none] (22) at (-14, -1) {};
		\node [style=none] (23) at (-12, -1) {};
		\node [style=none] (24) at (-13, 0) {$A$};
		\node [style=none] (29) at (-12, 0) {};
		\node [style=none] (31) at (-10.5, 0) {};
		\node [style=none] (33) at (-8.5, 0) {};
		\node [style=none] (35) at (-7.5, 0) {};
		\node [style=none] (37) at (-10.5, 1) {};
		\node [style=none] (38) at (-8.5, 1) {};
		\node [style=none] (39) at (-10.5, -1) {};
		\node [style=none] (40) at (-8.5, -1) {};
		\node [style=none] (41) at (-9.5, 0) {$B$};
		\node [style=none] (43) at (-4.5, 0) {};
		\node [style=none] (45) at (-3.5, 0) {};
		\node [style=none] (47) at (-1.5, 0) {};
		\node [style=none] (49) at (-0.5, 0) {};
		\node [style=none] (51) at (-3.5, 1) {};
		\node [style=none] (52) at (-1.5, 1) {};
		\node [style=none] (53) at (-3.5, -1) {};
		\node [style=none] (54) at (-1.5, -1) {};
		\node [style=none] (55) at (-2.5, 0) {$BA$};
		\node [style=none] (58) at (-13.5, -3) {};
		\node [style=none] (60) at (-12.5, -3) {};
		\node [style=none] (62) at (-10.5, -3) {};
		\node [style=none] (64) at (-9.5, -3) {};
		\node [style=none] (65) at (-12.5, -2) {};
		\node [style=none] (66) at (-10.5, -2) {};
		\node [style=none] (67) at (-12.5, -4) {};
		\node [style=none] (68) at (-10.5, -4) {};
		\node [style=none] (69) at (-11.5, -3) {$A$};
		\node [style=none] (72) at (-6, 0) {$=$};
		\node [style=none] (75) at (-13.5, -5.5) {};
		\node [style=none] (77) at (-12.5, -5.5) {};
		\node [style=none] (79) at (-10.5, -5.5) {};
		\node [style=none] (81) at (-9.5, -5.5) {};
		\node [style=none] (82) at (-12.5, -4.5) {};
		\node [style=none] (83) at (-10.5, -4.5) {};
		\node [style=none] (84) at (-12.5, -6.5) {};
		\node [style=none] (85) at (-10.5, -6.5) {};
		\node [style=none] (86) at (-11.5, -5.5) {$B$};
		\node [style=none] (88) at (-6.5, -3.25) {};
		\node [style=none] (89) at (-6.5, -5.25) {};
		\node [style=none] (90) at (-5.5, -3.25) {};
		\node [style=none] (91) at (-5.5, -5.25) {};
		\node [style=none] (92) at (-2, -3.25) {};
		\node [style=none] (93) at (-2, -5.25) {};
		\node [style=none] (94) at (-1, -3.25) {};
		\node [style=none] (95) at (-1, -5.25) {};
		\node [style=none] (96) at (-5.5, -2.75) {};
		\node [style=none] (97) at (-2, -2.75) {};
		\node [style=none] (98) at (-5.5, -5.75) {};
		\node [style=none] (99) at (-2, -5.75) {};
		\node [style=none] (100) at (-3.75, -4.25) {$\begin{pmatrix} A & 0\\ 0 & B \end{pmatrix}$};
		\node [style=none] (105) at (-7.75, -4.25) {$=$};
		\node [style=none] (108) at (-14.5, 0.5) {};
		\node [style=none] (109) at (-11.25, 0.5) {};
		\node [style=none] (110) at (-14.5, 0.5) {$\scriptstyle m$};
		\node [style=none] (111) at (-11.25, 0.5) { $\scriptstyle m'$};
		\node [style=none] (112) at (-7.75, 0.5) {$\scriptstyle m''$};
		\node [style=none] (113) at (-4.25, 0.5) {$\scriptstyle m$};
		\node [style=none] (114) at (-0.75, 0.5) {$\scriptstyle m''$};
		\node [style=none] (115) at (-13.25, -2.5) {$\scriptstyle m$};
		\node [style=none] (116) at (-9.75, -2.5) {$\scriptstyle m'$};
		\node [style=none] (117) at (-13.25, -5) {$\scriptstyle m'$};
		\node [style=none] (118) at (-9.75, -5) {$\scriptstyle m''$};
		\node [style=none] (119) at (-6.25, -2.75) {$\scriptstyle m$};
		\node [style=none] (120) at (-6.25, -4.75) {$\scriptstyle m'$};
		\node [style=none] (121) at (-1.25, -2.75) {$\scriptstyle m'$};
		\node [style=none] (122) at (-1.25, -4.75) {$\scriptstyle m''$};
	\end{pgfonlayer}
	\begin{pgfonlayer}{edgelayer}
		\draw (0.center) to (2.center);
		\draw (20.center) to (22.center);
		\draw (22.center) to (23.center);
		\draw (23.center) to (21.center);
		\draw (21.center) to (20.center);
		\draw (29.center) to (31.center);
		\draw (33.center) to (35.center);
		\draw (37.center) to (39.center);
		\draw (39.center) to (40.center);
		\draw (40.center) to (38.center);
		\draw (38.center) to (37.center);
		\draw (43.center) to (45.center);
		\draw (47.center) to (49.center);
		\draw (51.center) to (53.center);
		\draw (53.center) to (54.center);
		\draw (54.center) to (52.center);
		\draw (52.center) to (51.center);
		\draw (58.center) to (60.center);
		\draw (62.center) to (64.center);
		\draw (65.center) to (67.center);
		\draw (67.center) to (68.center);
		\draw (68.center) to (66.center);
		\draw (66.center) to (65.center);
		\draw (75.center) to (77.center);
		\draw (79.center) to (81.center);
		\draw (82.center) to (84.center);
		\draw (84.center) to (85.center);
		\draw (85.center) to (83.center);
		\draw (83.center) to (82.center);
		\draw (88.center) to (90.center);
		\draw (89.center) to (91.center);
		\draw (92.center) to (94.center);
		\draw (93.center) to (95.center);
		\draw (96.center) to (98.center);
		\draw (98.center) to (99.center);
		\draw (99.center) to (97.center);
		\draw (97.center) to (96.center);
	\end{pgfonlayer}
\end{tikzpicture}}$$
The labels on the input and output wires represent the underlying dimensions of a matrix with $m$ rows and $m'$ columns.
Two parallel wires with labels $m$ and $m'$ have dimension $m + m'$ and wires without a label have dimension $1$.
By composing beam splitters and phase shifts we always make a unitary operation and there are efficient algorithms to route any $m \times m$ unitary matrix as a circuit of phase shifts and beam splitters
\cite{reck_experimental_1994,clements_optimal_2016}.
Any matrix can be decomposed into a diagram built from the split, merge and box generators defined by 
$$\scalebox{0.7}{\begin{tikzpicture}
	\begin{pgfonlayer}{nodelayer}
		\node [style=none] (0) at (-1.25, 0) {};
		\node [style=lsplit] (1) at (0, 0) {};
		\node [style=none] (2) at (1.5, 0.75) {};
		\node [style=none] (3) at (1.5, -0.75) {};
		\node [style=none] (4) at (2.5, 0) {$=$};
		\node [style=none] (5) at (5.75, 0) {\Large $\begin{pmatrix} 1 \\ 1 \end{pmatrix}$};
		\node [style=none] (6) at (4.75, 1.5) {};
		\node [style=none] (7) at (6.75, 1.5) {};
		\node [style=none] (8) at (4.75, -1.5) {};
		\node [style=none] (9) at (6.75, -1.5) {};
		\node [style=none] (11) at (3.75, 0) {};
		\node [style=none] (12) at (7.75, 1) {};
		\node [style=none] (13) at (4.75, 0) {};
		\node [style=none] (14) at (6.75, 1) {};
		\node [style=none] (15) at (7.75, -1) {};
		\node [style=none] (16) at (6.75, -1) {};
	\end{pgfonlayer}
	\begin{pgfonlayer}{edgelayer}
		\draw [in=45, out=180] (2.center) to (1);
		\draw [in=-180, out=-45] (1) to (3.center);
		\draw (1) to (0.center);
		\draw (6.center) to (8.center);
		\draw (8.center) to (9.center);
		\draw (9.center) to (7.center);
		\draw (7.center) to (6.center);
		\draw (11.center) to (13.center);
		\draw (14.center) to (12.center);
		\draw (16.center) to (15.center);
	\end{pgfonlayer}
\end{tikzpicture}
} \, , \, \scalebox{0.7}{\begin{tikzpicture}
	\begin{pgfonlayer}{nodelayer}
		\node [style=none] (0) at (1.25, 0) {};
		\node [style=rmerge] (1) at (0, 0) {};
		\node [style=none] (2) at (-1.5, 0.75) {};
		\node [style=none] (3) at (-1.5, -0.75) {};
		\node [style=none] (4) at (2.25, 0) {$=$};
		\node [style=none] (5) at (6, 0) {\Large  $\begin{pmatrix} 1 & 1 \end{pmatrix}$};
		\node [style=none] (7) at (4.25, 1) {};
		\node [style=none] (8) at (7.75, 1) {};
		\node [style=none] (9) at (4.25, -1) {};
		\node [style=none] (10) at (7.75, -1) {};
		\node [style=none] (11) at (3.25, 0.5) {};
		\node [style=none] (12) at (4.25, -0.5) {};
		\node [style=none] (13) at (4.25, 0.5) {};
		\node [style=none] (14) at (3.25, -0.5) {};
		\node [style=none] (15) at (8.75, 0) {};
		\node [style=none] (16) at (7.75, 0) {};
	\end{pgfonlayer}
	\begin{pgfonlayer}{edgelayer}
		\draw [in=135, out=0] (2.center) to (1);
		\draw [in=0, out=-135] (1) to (3.center);
		\draw (1) to (0.center);
		\draw (7.center) to (9.center);
		\draw (9.center) to (10.center);
		\draw (10.center) to (8.center);
		\draw (8.center) to (7.center);
		\draw (11.center) to (13.center);
		\draw (14.center) to (12.center);
		\draw (16.center) to (15.center);
	\end{pgfonlayer}
\end{tikzpicture}
}$$ and $$\scalebox{0.7}{\begin{tikzpicture}
	\begin{pgfonlayer}{nodelayer}
		\node [style=none] (6) at (-1, 1) {};
		\node [style=none] (7) at (1, 1) {};
		\node [style=none] (8) at (-1, -1) {};
		\node [style=none] (9) at (1, -1) {};
		\node [style=none] (18) at (-2, 0) {};
		\node [style=none] (19) at (2, 0) {};
		\node [style=none] (20) at (-1, 0) {};
		\node [style=none] (21) at (1, 0) {};
		\node [style=none] (24) at (0, 0) {\Large $\begin{pmatrix} c \end{pmatrix}$};
	\end{pgfonlayer}
	\begin{pgfonlayer}{edgelayer}
		\draw (6.center) to (8.center);
		\draw (8.center) to (9.center);
		\draw (9.center) to (7.center);
		\draw (7.center) to (6.center);
		\draw (18.center) to (20.center);
		\draw (21.center) to (19.center);
	\end{pgfonlayer}
\end{tikzpicture}}$$ 
respectively, where $c \in \mathbb{C}$. For example, the beam splitter in equation \ref{eq-phase} can be written as
\begin{align}\label{eq-beamsplitter}
    \begin{split}
         \scalebox{0.85}{\begin{tikzpicture}
	\begin{pgfonlayer}{nodelayer}
		\node [style=none] (0) at (3.5, 1) {};
		\node [style=lsplit] (1) at (3.5, 1) {};
		\node [style=none] (2) at (6.5, 1) {};
		\node [style=none] (3) at (6.5, -1) {};
		\node [style=none] (4) at (7.5, 1) {};
		\node [style=rmerge] (5) at (6.5, 1) {};
		\node [style=none] (6) at (3.5, -1) {};
		\node [style=lsplit] (7) at (3.5, -1) {};
		\node [style=rmerge] (8) at (6.5, -1) {};
		\node [style=none] (9) at (7.5, -1) {};
		\node [style=none] (10) at (1, 1) {};
		\node [style=none] (11) at (1, -1) {};
		\node [style=phase] (12) at (5, 1.5) {$i$};
		\node [style=phase] (13) at (4.275, 0.475) {$1$};
		\node [style=phase] (14) at (4.25, -0.475) {$1$};
		\node [style=phase] (15) at (5, -1.5) {$i$};
		\node [style=none] (26) at (0, 0) {$=$};
		\node [style=none] (31) at (-6, 1.5) {};
		\node [style=none] (32) at (-1.5, 1.5) {};
		\node [style=none] (33) at (-6, -1.5) {};
		\node [style=none] (34) at (-1.5, -1.5) {};
		\node [style=none] (35) at (-3.75, 0) {$\frac{1}{\sqrt{2}}\begin{pmatrix}i & 1 \\ 1 & i \end{pmatrix}$};
		\node [style=none] (37) at (-6.75, 1) {};
		\node [style=none] (38) at (-0.75, 1) {};
		\node [style=none] (39) at (-6, 1) {};
		\node [style=none] (40) at (-1.5, 1) {};
		\node [style=none] (41) at (-6.75, -1) {};
		\node [style=none] (42) at (-6, -1) {};
		\node [style=none] (43) at (-0.75, -1) {};
		\node [style=none] (44) at (-1.5, -1) {};
		\node [style=phase] (45) at (2, 1) {$\frac{1}{\sqrt{2}}$};
		\node [style=phase] (47) at (2, -1) {$\frac{1}{\sqrt{2}}$};
	\end{pgfonlayer}
	\begin{pgfonlayer}{edgelayer}
		\draw (5) to (4.center);
		\draw (10.center) to (1);
		\draw (11.center) to (7);
		\draw (8) to (9.center);
		\draw [bend left=15] (1) to (12);
		\draw [bend left=15, looseness=0.75] (12) to (5);
		\draw (1) to (13);
		\draw (13) to (8);
		\draw (7) to (14);
		\draw (14) to (5);
		\draw [bend right=15] (7) to (15);
		\draw [in=210, out=0] (15) to (8);
		\draw (31.center) to (33.center);
		\draw (33.center) to (34.center);
		\draw (34.center) to (32.center);
		\draw (32.center) to (31.center);
		\draw (37.center) to (39.center);
		\draw (40.center) to (38.center);
		\draw (41.center) to (42.center);
		\draw (44.center) to (43.center);
	\end{pgfonlayer}
\end{tikzpicture}}.
    \end{split}
\end{align}

\subsection{Photons}
The quantum features of linear optics arise when we allow for creation and annihilation of indistinguishable particles. 
We do this -- syntactically -- by using $n$-photon states $\begin{tikzpicture}
	\begin{pgfonlayer}{nodelayer}
		\node [style=none] (0) at (0.75, 0) {};
		\node [style=black node] (1) at (-0.5, 0) {};
		\node [style=none] (2) at (-0.5, 0.5) {$\scriptstyle n$};
	\end{pgfonlayer}
	\begin{pgfonlayer}{edgelayer}
		\draw (1) to (0.center);
	\end{pgfonlayer}
\end{tikzpicture}
$ and effects $\begin{tikzpicture}
	\begin{pgfonlayer}{nodelayer}
		\node [style=none] (0) at (-0.5, 0) {};
		\node [style=black node] (1) at (0.75, 0) {};
		\node [style=none] (2) at (0.75, 0.5) {$\scriptstyle n$};
	\end{pgfonlayer}
	\begin{pgfonlayer}{edgelayer}
		\draw (1) to (0.center);
	\end{pgfonlayer}
\end{tikzpicture}
$. The behaviour of $n$-photon diagrams can be summarised with the equations
\begin{align}\label{eq-photon-defs}
    \begin{split}
         \scalebox{0.8}{\begin{tikzpicture}
	\begin{pgfonlayer}{nodelayer}
		\node [style=none] (49) at (-1, 2.5) {};
		\node [style=black node] (50) at (-3.5, 2.5) {};
		\node [style=none] (51) at (-3.5, 3) {$n$};
		\node [style=black node] (56) at (-1, 2.5) {};
		\node [style=none] (57) at (-1, 3) {$n'$};
		\node [style=none] (58) at (0, 2.5) {$=$};
		\node [style=phase] (59) at (-2.25, 2.475) {$c$};
		\node [style=none] (60) at (1.75, 2.5) {$c^n \, \delta_{n, n'}$};
		\node [style=black node] (61) at (-3.5, -2) {};
		\node [style=black node] (62) at (-3.5, -3.5) {};
		\node [style=black node] (63) at (2.875, -2.75) {};
		\node [style=none] (64) at (4.875, -2.75) {};
		\node [style=none] (65) at (2.875, -2.25) {$n$};
		\node [style=none] (66) at (-1, -2.75) {};
		\node [style=none] (67) at (-3.5, -2.5) {$\vdots$};
		\node [style=none] (68) at (-4, -2.75) {$n$};
		\node [style=none] (69) at (0, -2.75) {$=$};
		\node [style=none] (70) at (1.375, -2.725) {$\sqrt{n!}$};
		\node [style=none] (71) at (-3.5, -4) {$1$};
		\node [style=none] (72) at (-3.5, -1.5) {$1$};
		\node [style=rmerge] (73) at (-2.375, -2.75) {};
		\node [style=none] (74) at (-3.375, -3.5) {};
		\node [style=none] (75) at (-3.375, -2) {};
		\node [style=none] (76) at (2.5, 0.75) {};
		\node [style=none] (77) at (2.5, -0.75) {};
		\node [style=none] (78) at (0, 0) {$=$};
		\node [style=black node] (79) at (-3.5, 0) {};
		\node [style=none] (80) at (-3.5, 0.5) {$1$};
		\node [style=black node] (81) at (1.25, -0.75) {};
		\node [style=none] (82) at (1.25, -0.25) {$0$};
		\node [style=black node] (83) at (1.25, 0.75) {};
		\node [style=none] (84) at (1.25, 1.25) {$1$};
		\node [style=lsplit] (85) at (-2.25, 0) {};
		\node [style=none] (86) at (-1.25, -0.75) {};
		\node [style=none] (87) at (-1.25, 0.75) {};
		\node [style=none] (88) at (-1, -0.75) {};
		\node [style=none] (89) at (-1, 0.75) {};
		\node [style=none] (90) at (6, 0.75) {};
		\node [style=none] (91) at (6, -0.75) {};
		\node [style=black node] (92) at (4.75, -0.75) {};
		\node [style=none] (93) at (4.75, -0.25) {$1$};
		\node [style=black node] (94) at (4.75, 0.75) {};
		\node [style=none] (95) at (4.75, 1.25) {$0$};
		\node [style=none] (96) at (3.5, 0) {$+$};
	\end{pgfonlayer}
	\begin{pgfonlayer}{edgelayer}
		\draw (50) to (59);
		\draw (59) to (56);
		\draw (63) to (64.center);
		\draw [in=0, out=-135, looseness=0.75] (73) to (74.center);
		\draw [in=135, out=0, looseness=0.75] (75.center) to (73);
		\draw (73) to (66.center);
		\draw (75.center) to (61);
		\draw (74.center) to (62);
		\draw (81) to (77.center);
		\draw (83) to (76.center);
		\draw [in=180, out=-45, looseness=0.75] (85) to (86.center);
		\draw [in=45, out=180, looseness=0.75] (87.center) to (85);
		\draw (89.center) to (87.center);
		\draw (88.center) to (86.center);
		\draw (79) to (85);
		\draw (92) to (91.center);
		\draw (94) to (90.center);
	\end{pgfonlayer}
\end{tikzpicture}}
    \end{split}
\end{align}
where $c$ is a scalar and $\delta$ is the Dirac delta. 
The notation where multiple wires are entering a merge is justified by the associativity of the merge. 
These rules also apply when we take the vertical reflection which corresponds to taking the transpose.
Any matrix with $n$-photon states and effects, also called QPath diagram~\cite{de_felice_quantum_2023}, can be rewritten in a normal form $(A, I, J, \alpha)$ 
where $A$ is an $(m + k) \times (m' + k)'$ matrix, $I\in \mathbb{N}^{k}$ and $J \in \mathbb{N}^{k'}$ are basis states within the Fock space and $\alpha$ is a complex scalar
$$\scalebox{1}
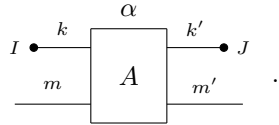
{\begin{tikzpicture}
	\begin{pgfonlayer}{nodelayer}
		\node [style=none] (1) at (-3, -0.75) {};
		\node [style=none] (3) at (-1, -0.75) {};
		\node [style=none] (5) at (1, -0.75) {};
		\node [style=none] (7) at (3, -0.75) {};
		\node [style=none] (8) at (-2.5, 0.75) {};
		\node [style=none] (10) at (-1, 0.75) {};
		\node [style=none] (12) at (1, 0.75) {};
		\node [style=none] (14) at (2.5, 0.75) {};
		\node [style=black node] (16) at (-2.5, 0.75) {};
		\node [style=black node] (18) at (2.5, 0.75) {};
		\node [style=none] (20) at (-1, 1.25) {};
		\node [style=none] (21) at (1, 1.25) {};
		\node [style=none] (22) at (-1, -1.25) {};
		\node [style=none] (23) at (1, -1.25) {};
		\node [style=none] (24) at (0, 0) {$A$};
		\node [style=none] (27) at (3, 0.75) {$\scriptstyle J$};
		\node [style=none] (28) at (-3, 0.75) {$\scriptstyle I$};
		\node [style=none] (30) at (-2, -0.25) {$\scriptstyle m$};
		\node [style=none] (31) at (2, -0.25) {$\scriptstyle m'$};
		\node [style=none] (32) at (-1.75, 1.25) {$\scriptstyle k$};
		\node [style=none] (33) at (1.75, 1.25) {$\scriptstyle k'$};
		\node [style=none] (34) at (0, 1.75) {$\alpha$};
	\end{pgfonlayer}
	\begin{pgfonlayer}{edgelayer}
		\draw (1.center) to (3.center);
		\draw (5.center) to (7.center);
		\draw (8.center) to (10.center);
		\draw (12.center) to (14.center);
		\draw (20.center) to (22.center);
		\draw (22.center) to (23.center);
		\draw (23.center) to (21.center);
		\draw (21.center) to (20.center);
	\end{pgfonlayer}
\end{tikzpicture}} \, .$$

In linear optics we have a single-particle unitary $U$ which extends to a ``non-interacting'' operator $\widetilde{U}$ on the multi-particle space (symmetric Fock space).
This mapping extends to any matrix $A$ giving an operator $\widetilde{A}$, such that $\widetilde{AB} = \widetilde{A}\widetilde{B}$, $\widetilde{A \oplus B} = \widetilde{A} \otimes \widetilde{B}$
and $\mathbb{I} = \widetilde{\mathtt{I}}$ where $\mathtt{I}$ and $\mathbb{I}$ are the identity on the single-particle 
and multi-particle space, respectively.
Multiple definitions of this mapping are available in the literature \cite{vicary_categorical_2008, aaronson_computational_2010,ivanov_complexity_2020}.
Here we use a combinatorial definition due to Aaronson \cite{aaronson_computational_2010}:
the action of the operator $\widetilde{A}$ on the Fock space is defined by the amplitudes
for input and output basis states $X, Y \in \mathbb{N}^m$, given by
\begin{equation}\label{eq-amplitudes}
	\bra{Y} \widetilde{A} \ket{X} \; = \begin{cases}
   		  				           \frac{\mathtt{Perm}(A_{X, Y})} {\sqrt{\prod_{x} X_x! \prod_y Y_y!}}& \text{if $n_X = n_Y$} \\
							   0 & \text{otherwise}
						       \end{cases}
\end{equation}
where $\mathtt{Perm}$ denotes the matrix \emph{permanent}, $n_X = \sum_{i = 1}^m X_i$ is the number of photons in $X$, and $A_{X, Y}$ is obtained as follows.
First we construct the $m \times n_Y$ matrix $A_Y$ by taking $Y_j$ copies of the $j$th column of
$A$ for each $j \leq m$, then we construct $A_{X, Y}$ by taking $X_i$ copies of the $i$th row of $A_Y$.
This process corresponds to the application of the last rule in (\ref{eq-photon-defs}). 
The effect of the introduced merge and split maps is to copy the rows and columns of $A$. 
The mapping moreover extends to arbitrary diagrams $(A, I, J,\alpha)$, 
giving operators of the form $\alpha (\bra{J} \otimes \mathbb{I})\widetilde{A}(\ket{I} \otimes \mathbb{I})$.

When all input and output wires of a diagram are capped with states and effects, we can reduce the diagram to a scalar using the rules above.
If the diagram represents a linear optical circuit, the scalar represents the \emph{amplitude} of measuring the output effect for the given input state.
Using (\ref{eq-beamsplitter}) and (\ref{eq-photon-defs}), we demonstrate the HOM-dip effect where two indistinguishable photons entering a beam splitter from either input will never leave separately from either output.
\begin{align}\label{eq-hom-dip}
    \begin{split}
         \scalebox{0.8}{\begin{tikzpicture}
	\begin{pgfonlayer}{nodelayer}
		\node [style=none] (22) at (-15, -1.25) {};
		\node [style=none] (23) at (-15, 1) {};
		\node [style=none] (24) at (-11, 1) {};
		\node [style=none] (25) at (-11, -1.25) {};
		\node [style=none] (26) at (-9, 0) {$=$};
		\node [style=black node] (28) at (-11, 1) {};
		\node [style=black node] (29) at (-11, -1.25) {};
		\node [style=black node] (30) at (-15, 1) {};
		\node [style=black node] (31) at (-15, -1.25) {};
		\node [style=none] (36) at (-15, -4.25) {$=$};
		\node [style=none] (37) at (-13, -5.25) {};
		\node [style=none] (38) at (-13, -3.25) {};
		\node [style=none] (39) at (-10, -3.25) {};
		\node [style=none] (40) at (-10, -5.25) {};
		\node [style=black node] (41) at (-10, -3.25) {};
		\node [style=black node] (42) at (-10, -5.25) {};
		\node [style=black node] (43) at (-13, -3.25) {};
		\node [style=black node] (44) at (-13, -5.25) {};
		\node [style=none] (45) at (-7, -5.25) {};
		\node [style=none] (46) at (-7, -3.25) {};
		\node [style=black node] (49) at (-4, -5.25) {};
		\node [style=black node] (50) at (-4, -3.25) {};
		\node [style=black node] (51) at (-7, -3.25) {};
		\node [style=black node] (52) at (-7, -5.25) {};
		\node [style=none] (53) at (-8.5, -4.25) {$+$};
		\node [style=phase] (57) at (-11.5, -3.25) {$i$};
		\node [style=phase] (58) at (-11.5, -5.25) {$i$};
		\node [style=none] (59) at (-2, -4.25) {$=$};
		\node [style=none] (60) at (-0.5, -4.25) {$0$};
		\node [style=none] (61) at (-15, 1) {};
		\node [style=none] (62) at (-11, 1) {};
		\node [style=none] (63) at (-15, -1.25) {};
		\node [style=none] (64) at (-11, -1.25) {};
		\node [style=none] (65) at (-14, 0) {};
		\node [style=none] (66) at (-12, 0) {};
		\node [style=none] (67) at (-14, -0.25) {};
		\node [style=none] (68) at (-12, -0.25) {};
		\node [style=none] (69) at (-13, 0) {};
		\node [style=none] (70) at (-13, -0.25) {};
		\node [style=none] (71) at (-6, 1) {};
		\node [style=lsplit] (72) at (-6, 1) {};
		\node [style=none] (73) at (-3, 1) {};
		\node [style=none] (74) at (-3, -1) {};
		\node [style=none] (75) at (-2, 1) {};
		\node [style=rmerge] (76) at (-3, 1) {};
		\node [style=none] (77) at (-6, -1) {};
		\node [style=lsplit] (78) at (-6, -1) {};
		\node [style=rmerge] (79) at (-3, -1) {};
		\node [style=none] (80) at (-2, -1) {};
		\node [style=none] (81) at (-7, 1) {};
		\node [style=none] (82) at (-7, -1) {};
		\node [style=phase] (83) at (-4.5, 1.5) {$i$};
		\node [style=phase] (84) at (-5.225, 0.475) {$1$};
		\node [style=phase] (85) at (-5.25, -0.475) {$1$};
		\node [style=phase] (86) at (-4.5, -1.5) {$i$};
		\node [style=none] (87) at (-7.75, 0) {$\frac{1}{2}$};
		\node [style=black node] (88) at (-2, -1) {};
		\node [style=black node] (89) at (-2, 1) {};
		\node [style=black node] (90) at (-7, 1) {};
		\node [style=black node] (91) at (-7, -1) {};
	\end{pgfonlayer}
	\begin{pgfonlayer}{edgelayer}
		\draw (51) to (49);
		\draw (52) to (50);
		\draw (43) to (57);
		\draw (57) to (41);
		\draw (44) to (58);
		\draw (58) to (42);
		\draw (65.center) to (67.center);
		\draw (67.center) to (68.center);
		\draw (68.center) to (66.center);
		\draw (66.center) to (65.center);
		\draw [in=120, out=0] (61.center) to (69.center);
		\draw [in=-180, out=45] (69.center) to (62.center);
		\draw [in=180, out=-45] (70.center) to (64.center);
		\draw [in=0, out=-135] (70.center) to (63.center);
		\draw (76) to (75.center);
		\draw (81.center) to (72);
		\draw (82.center) to (78);
		\draw (79) to (80.center);
		\draw [bend left=15] (72) to (83);
		\draw [bend left=15, looseness=0.75] (83) to (76);
		\draw (72) to (84);
		\draw (84) to (79);
		\draw (78) to (85);
		\draw (85) to (76);
		\draw [bend right=15] (78) to (86);
		\draw [in=210, out=0] (86) to (79);
	\end{pgfonlayer}
\end{tikzpicture}}
    \end{split}
\end{align}
The amplitudes are expressed as sums over the matchings of photon states to photon effects.
More precisely, the amplitude of an $m \times m'$ matrix $A$ for a pair of basis states $X \in \mathbb{N}^m$ and $Y \in \mathbb{N}^{m'}$
is given by equation (\ref{eq-amplitudes}).

Diagrams with states and effects allow us to reason about bosonic operators that can not be represented as linear optical circuits.
The bosonic creation $a^\dagger$ and annihilation $a$ operators are given by $a^\dagger=\scalebox{0.7}{\begin{tikzpicture}
	\begin{pgfonlayer}{nodelayer}
		\node [style=black node] (2) at (-1.5, 0.75) {};
		\node [style=none] (4) at (1, 0) {};
		\node [style=rmerge] (5) at (-0.25, 0) {};
		\node [style=none] (6) at (-1.5, 0.75) {};
		\node [style=none] (7) at (-1.75, -0.75) {};
	\end{pgfonlayer}
	\begin{pgfonlayer}{edgelayer}
		\draw [in=0, out=-135] (5) to (7.center);
		\draw (5) to (4.center);
		\draw (5) to (6.center);
	\end{pgfonlayer}
\end{tikzpicture}}$ and $a=\scalebox{0.7}{\begin{tikzpicture}
	\begin{pgfonlayer}{nodelayer}
		\node [style=black node] (7) at (1.25, 0.75) {};
		\node [style=none] (9) at (-1.25, 0) {};
		\node [style=lsplit] (10) at (0, 0) {};
		\node [style=none] (11) at (1.25, 0.75) {};
		\node [style=none] (12) at (1.5, -0.75) {};
	\end{pgfonlayer}
	\begin{pgfonlayer}{edgelayer}
		\draw [in=-180, out=-45] (10) to (12.center);
		\draw (10) to (9.center);
		\draw (10) to (11.center);
	\end{pgfonlayer}
\end{tikzpicture}}$. 
We obtain different representations of the number operator $\hat{n} = a^\dagger a$ 
\begin{align}\label{eq-number-op}
    \begin{split}
         \scalebox{0.8}{\begin{tikzpicture}
	\begin{pgfonlayer}{nodelayer}
		\node [style=none] (0) at (-5.5, 0) {};
		\node [style=rmerge] (1) at (-2, 0) {};
		\node [style=black node] (2) at (-3, 0.75) {};
		\node [style=none] (3) at (-1, 0) {};
		\node [style=lsplit] (6) at (-5.5, 0) {};
		\node [style=black node] (7) at (-4.5, 0.75) {};
		\node [style=none] (8) at (-6.5, 0) {};
		\node [style=none] (20) at (-8, 0) {$=$};
		\node [style=none] (21) at (-3.25, -1.75) {};
		\node [style=none] (22) at (-3.25, -4.25) {};
		\node [style=none] (23) at (-0.25, -4.25) {};
		\node [style=none] (24) at (-0.25, -1.75) {};
		\node [style=none] (25) at (-4.25, -3.75) {};
		\node [style=none] (26) at (0.75, -3.75) {};
		\node [style=none] (27) at (-3.25, -3.75) {};
		\node [style=none] (28) at (-0.25, -3.75) {};
		\node [style=none] (29) at (-1.75, -3) {$\begin{pmatrix} 0 & 1\\ 1 &1 \end{pmatrix}$};
		\node [style=black node] (30) at (0.75, -2.25) {};
		\node [style=black node] (32) at (-4.25, -2.25) {};
		\node [style=none] (34) at (-3.25, -2.25) {};
		\node [style=none] (35) at (-0.25, -2.25) {};
		\node [style=none] (36) at (-11, -2) {};
		\node [style=lsplit] (37) at (-11, -2) {};
		\node [style=none] (38) at (-8, -2) {};
		\node [style=none] (39) at (-8, -4) {};
		\node [style=none] (40) at (-7, -2) {};
		\node [style=rmerge] (41) at (-8, -2) {};
		\node [style=none] (42) at (-11, -4) {};
		\node [style=lsplit] (43) at (-11, -4) {};
		\node [style=rmerge] (44) at (-8, -4) {};
		\node [style=none] (45) at (-7, -4) {};
		\node [style=none] (46) at (-12, -2) {};
		\node [style=none] (47) at (-12, -4) {};
		\node [style=phase] (48) at (-9.5, -1.5) {$0$};
		\node [style=phase] (49) at (-10.225, -2.525) {$1$};
		\node [style=phase] (50) at (-10.25, -3.475) {$1$};
		\node [style=phase] (51) at (-9.5, -4.5) {$1$};
		\node [style=none] (53) at (-5.75, -3) {$=$};
		\node [style=black node] (54) at (-7, -2) {};
		\node [style=black node] (55) at (-12, -2) {};
		\node [style=phase] (56) at (-11, 0) {$\hat{n}$};
		\node [style=none] (57) at (-9.5, 0) {};
		\node [style=none] (58) at (-12.5, 0) {};
		\node [style=none] (59) at (-13.25, -3) {$=$};
	\end{pgfonlayer}
	\begin{pgfonlayer}{edgelayer}
		\draw [bend right] (0.center) to (1);
		\draw (1) to (3.center);
		\draw (2) to (1);
		\draw (6) to (8.center);
		\draw (7) to (6);
		\draw (25.center) to (27.center);
		\draw (21.center) to (22.center);
		\draw (22.center) to (23.center);
		\draw (23.center) to (24.center);
		\draw (24.center) to (21.center);
		\draw (28.center) to (26.center);
		\draw (32) to (34.center);
		\draw (35.center) to (30);
		\draw (41) to (40.center);
		\draw (46.center) to (37);
		\draw (47.center) to (43);
		\draw (44) to (45.center);
		\draw [bend left=15] (37) to (48);
		\draw [bend left=15, looseness=0.75] (48) to (41);
		\draw (37) to (49);
		\draw (49) to (44);
		\draw (43) to (50);
		\draw (50) to (41);
		\draw [bend right=15] (43) to (51);
		\draw [in=210, out=0] (51) to (44);
		\draw (56) to (57.center);
		\draw (56) to (58.center);
	\end{pgfonlayer}
\end{tikzpicture}}
    \end{split}
\end{align}
The latter gives the underlying matrix $W = \begin{pmatrix} 0 & 1\\ 1 &1 \end{pmatrix}$.
It has spectral norm $\norm{W} = \frac{1+\sqrt{5}}{2}$ greater than $1$. 
This can be used to prove that the number operator is an unbounded operator on the Fock space. 
Indeed $\widetilde{A}$ is bounded if and only if $\norm{A} \leq 1$.
We see that the number operator is not bounded nor unitary so can not be physically implemented in its current form.

\subsection{Dilation}
Heralded linear optical circuits --- built from beam splitters, phase shifts, photon sources and detectors --- 
realise, by post-selection, the class of bounded operators of the form $(U, I, J)$ where $U$ is a unitary matrix.
To implement $(U, I, J)$ we have to prepare an ancillary state $I$ in the first $k$ input modes
and post-select for the output effect $J$ on the first $k'$ output modes.

Dilations allow us to simulate arbitrary diagrams with a heralded linear optical circuit.
A unitary dilation for a finite-dimensional matrix $A$ is a larger unitary matrix $U_A$ that contains $A$ as a block.
In the context of linear optics, dilations have been used to encode combinatorial quantities into the amplitudes of a unitary circuit \cite{mezher_solving_2023}. 
Generalising this work, we show that heralded circuits are sufficient to simulate arbitrary matrices with
$n$-photon states and effects.

\begin{theorem}\label{thm-dilation}
    For any diagram $(A, I, J)$, there is a unitary matrix $U_A$
    such that the following are equal as operators:
    $$\scalebox{0.9}{\begin{tikzpicture}
	\begin{pgfonlayer}{nodelayer}
		\node [style=none] (34) at (1.5, -0.75) {};
		\node [style=none] (36) at (4, -0.75) {};
		\node [style=none] (38) at (7, -0.75) {};
		\node [style=none] (40) at (9.5, -0.75) {};
		\node [style=none] (42) at (2, 0.75) {};
		\node [style=none] (44) at (4, 0.75) {};
		\node [style=none] (46) at (7, 0.75) {};
		\node [style=none] (48) at (9, 0.75) {};
		\node [style=black node] (50) at (2, 0.75) {};
		\node [style=black node] (52) at (9, 0.75) {};
		\node [style=none] (53) at (4, 2.75) {};
		\node [style=none] (54) at (7, 2.75) {};
		\node [style=none] (55) at (4, -1.25) {};
		\node [style=none] (56) at (7, -1.25) {};
		\node [style=none] (57) at (5.5, 0.75) {$s_A \,  U_A$};
		\node [style=none] (58) at (9.5, 0.75) {$\scriptstyle J$};
		\node [style=none] (61) at (1.5, 0.75) {$\scriptstyle I$};
		\node [style=none] (66) at (0, 0) {$=$};
		\node [style=none] (67) at (2, 2.25) {};
		\node [style=black node] (69) at (2, 2.25) {};
		\node [style=none] (71) at (1.5, 2.25) {$\scriptstyle \vec{0}$};
		\node [style=none] (73) at (9, 2.25) {};
		\node [style=black node] (75) at (9, 2.25) {};
		\node [style=none] (77) at (9.5, 2.25) {$\scriptstyle \vec{0}$};
		\node [style=none] (82) at (3, 2.75) {$\scriptstyle m + k$};
		\node [style=none] (83) at (4, 2.25) {};
		\node [style=none] (85) at (7, 2.25) {};
		\node [style=none] (95) at (8, 2.75) {$\scriptstyle m' + k'$};
		\node [style=none] (96) at (-7.5, -0.75) {};
		\node [style=none] (97) at (-5.5, -0.75) {};
		\node [style=none] (98) at (-3.5, -0.75) {};
		\node [style=none] (99) at (-1.5, -0.75) {};
		\node [style=none] (100) at (-7, 0.75) {};
		\node [style=none] (101) at (-5.5, 0.75) {};
		\node [style=none] (102) at (-3.5, 0.75) {};
		\node [style=none] (103) at (-2, 0.75) {};
		\node [style=black node] (104) at (-7, 0.75) {};
		\node [style=black node] (105) at (-2, 0.75) {};
		\node [style=none] (106) at (-5.5, 1.25) {};
		\node [style=none] (107) at (-3.5, 1.25) {};
		\node [style=none] (108) at (-5.5, -1.25) {};
		\node [style=none] (109) at (-3.5, -1.25) {};
		\node [style=none] (110) at (-4.5, 0) {$A$};
		\node [style=none] (111) at (-1.5, 0.75) {$\scriptstyle J$};
		\node [style=none] (112) at (-7.5, 0.75) {$\scriptstyle I$};
		\node [style=none] (113) at (-6.5, -0.25) {$\scriptstyle m$};
		\node [style=none] (114) at (-2.5, -0.25) {$\scriptstyle m'$};
		\node [style=none] (115) at (-6.25, 1.25) {$\scriptstyle k$};
		\node [style=none] (116) at (-2.75, 1.25) {$\scriptstyle k'$};
		\node [style=none] (117) at (3, 1.25) {$\scriptstyle k$};
		\node [style=none] (118) at (8, 1.25) {$\scriptstyle k'$};
		\node [style=none] (120) at (2.75, -0.25) {$\scriptstyle m$};
		\node [style=none] (121) at (8.25, -0.25) {$\scriptstyle m'$};
	\end{pgfonlayer}
	\begin{pgfonlayer}{edgelayer}
		\draw (34.center) to (36.center);
		\draw (38.center) to (40.center);
		\draw (42.center) to (44.center);
		\draw (46.center) to (48.center);
		\draw (53.center) to (55.center);
		\draw (55.center) to (56.center);
		\draw (56.center) to (54.center);
		\draw (54.center) to (53.center);
		\draw (69) to (83.center);
		\draw (85.center) to (75);
		\draw (96.center) to (97.center);
		\draw (98.center) to (99.center);
		\draw (100.center) to (101.center);
		\draw (102.center) to (103.center);
		\draw (106.center) to (108.center);
		\draw (108.center) to (109.center);
		\draw (109.center) to (107.center);
		\draw (107.center) to (106.center);
	\end{pgfonlayer}
\end{tikzpicture}}$$
    Equivalently, for any input state $X\in \mathbb{N}^{m}$ and output effect $Y \in \mathbb{N}^{m'}$: 
    $$\bra{J, Y} A \ket{I, X} = s_A^n \bra{\vec{0}_{m+k}, J, Y} U_A \ket{\vec{0}_{m'+k'}, I, X}$$
    where  $\ket{I, X}$ is the concatenation of states $I$ and $X$, $\vec{0}_g$ is the input state of 0 photons on $g$ modes, $n = n_I + n_X$ is the
number of input photons, $s_A = \norm{A}$ if $\norm{A} \geq 1$ otherwise $s_A = 1$ 
    and $\norm{A}$ denotes the spectral norm of $A$, i.e. its largest singular value. 
\end{theorem}
\begin{proof}
    Since $\norm{A}$ is finite, the matrix $T = \frac{A}{s_A}$ is always a contraction. 
	It therefore admits a unitary dilation $U_A$.
	We may construct such a dilation following Halmos \cite{halmos1950}.
	First, we compute the defect operator $D(T) = \sqrt{\mathtt{I} - T T^\dagger}$ using the singular value decomposition of $A$.
	The Halmos dilation $U_A$ is then given by the unitary operator
	\begin{equation}\label{eqn:unitary}
	    U_A =
	    \begin{pmatrix}
	    - T^\dagger & D(T^\dagger)\\ 
	    D(T) & T
	    \end{pmatrix}.
	\end{equation}
	We see that
	\begin{align*}
	    \bra{\vec{0}_{m+k}, J, Y} U_A \ket{\vec{0}_{m'+k'}, I, X} =  \bra{J, Y}T \ket{I, X}\\
	        = \frac{1}{s_A^n}\bra{J, Y} A \ket{I, X}
	\end{align*}
	where the second equality with $n = n_I + n_X$ follows from (\ref{eq-amplitudes}) and a simple property of permanents:
	for any $n \times n$ matrix $X$ and complex number $z$, it holds that $\mathtt{Perm}(z X) = z^n \, \mathtt{Perm}(X)$.
\end{proof}

When the underlying matrix is a contraction, $\norm{A} \leq 1$, we may directly extend $A$ to a unitary on $2m$ modes. 
If instead $\norm{A} \geq 1$, the operator is unbounded and 
we need to normalise the amplitudes by a factor of $\norm{A}^n$ where $n$ is the number of input photons.
In Appendix \ref{app:dil_num_op}, we give an example showing how to simulate the number operator with a heralded linear optical circuit.

\section{Observables in linear optics}\label{sec:observables}

\subsection{Definition}

Consider an experimental setup where a state $\psi$ is prepared
and it evolves through an interferometer defined by the single-particle unitary $U$.
We sample from a distribution of photon numbers $X \in \mathbb{N}^m$ with the following probabilities: 
\begin{equation}\label{def-probabilities}
    P_U(X \vert \psi) = \Vert\bra{X} \widetilde{U} \ket{\psi}\Vert^2
\end{equation}
We may then consider different statistics or observables that can be computed from this distribution.

Throughout this paper, the state $\psi$ is assumed to be a pure state with a fixed number of photons.
More precisely, $\psi$ is a class of states parametrised by the number of modes $m$ and photons $n$,
given for any $m$ and $n$ by a superposition $\ket{\psi} = \sum_{X \in \Phi_{m, n}} \psi_X \ket{X}$ over the finite set
$\Phi_{m, n} = \{ X \in \mathbb{N}^m \, \vert \, n_X = \sum_i X_i = n \}$.
The natural example with $m = n$ is the bosonic product state $\ket{\psi} = \ket{1, \dots, 1}$,
which serves as initial state in the boson sampling proposal~\cite{aaronson_computational_2011}.
As we will see, different classical approximation results are known in this case, for estimation (rather than sampling) problems.
A second class of states of interest in our setting is given by entangled photonic states
that can be generated by quantum dot sources~\cite{istrati_sequential_2020, wein_photon-number_2022, thomas_efficient_2022}.
For example, linear photonic cluster states~\cite{istrati_sequential_2020} are $n$-photon states on $2n$ modes
of the form $\ket{\psi} = \frac{1}{\sqrt{2}^n}\sum_{b \in \{0, 1\}^n} \ket{X_b}$ where $(X_b)_{2i} = b_i$ and $(X_b)_{2i + 1} = \neg b_i$.
As the number of non-zero terms grows exponentially with $n$, the resulting distributions $\widetilde{U}\ket{\psi}$
may not be easy to simulate with bosonic product states.
In fermionic linear optics, advantage results are known for different classes of entangled initial states~\cite{ivanov_computational_2017, oszmaniec_fermion_2022},
while little is known in the bosonic setting \cite{ivanov_complexity_2020}.

We define a \emph{discrete observable} to be a normal operator $\mathcal{Q}$ on the Fock space that satisfies
$\mathcal{Q} = \widetilde{U}^\dagger \Lambda \widetilde{U}$ for an $m \times m$ unitary matrix $U$ and a diagonal operator $\Lambda$, 
in the sense that $\Lambda \ket{X} = \lambda(X) \ket{X}$ for any basis state $X \in \mathbb{N}^m$, 
where $\lambda: \mathbb{N}^m \to \mathbb{C}$ is a function assigning an eigenvalue to each basis vector.
In this case, we can evaluate the \emph{expectation values} of $\mathcal{Q}$ from the probabilities in (\ref{def-probabilities}). 
For any state $\psi$ with a finite number of photons $n$, we have:
\begin{align}\label{eq-expectation-values}
    E(\mathcal{Q}) &= \bra{\psi} \mathcal{Q} \ket{\psi} \nonumber \\
	 &= \bra{\psi} \widetilde{U}^\dagger \Lambda \widetilde{U} \ket{\psi} = \sum_{X \in \Phi_{m, n}} \lambda(X) P_U(X \vert \psi)
\end{align}
Note that $\mathcal{Q}$ is defined to be normal rather than self-adjoint (or Hermitian), 
this means the eigenfunction $\lambda$ is complex- rather than real-valued. 
We will see different examples of naturally occuring observables that are not self-adjoint.
By definition, any discrete observable is normal but the opposite is not true as there are 
unitaries on the Fock space that are not of the form $\widetilde{U}$.
Discrete observables represent the quantum statistics that we may compute given an initial state $\psi$ 
a passive linear optical circuit $U$ and number-resolving detectors.

\subsection{Sampling complexity}

A natural question is: can we estimate the expectation values in (\ref{eq-expectation-values}) \emph{efficiently} by sampling the photon-number distribution?
Sampling from $\tilde{U} \ket{\psi}$ gives us an obvious estimator for the expectation: we draw $T$ independent samples $\{ X_i \}_{i=1}^T$ and compute the sample mean
\begin{equation}
\overline{E}(\mathcal{Q}) = \frac{1}{T} \sum_{i=1}^T \lambda(X_i) \, .
\end{equation}
The question becomes: how big does $T$ need to be for an accurate estimation?
The answer to this question depends on the magnitude of $\lambda(X)$ over the states reachable from $\psi$. 
For example, if an eigenvalue is very large and its corresponding probability is small, 
then we will need more samples from the distribution to correctly account for it.

\begin{theorem}\label{thm-estimation}
    The expectation value $\bra{\psi} \mathcal{Q} \ket{\psi}$ of any discrete observable $\mathcal{Q}$ 
    can be estimated by sampling from $\widetilde{U}\ket{\psi}$, for a fixed success probability, in time $O(\frac{1}{\epsilon^2})$ 
    with an additive error of $\pm \epsilon \, s_{n, \lambda}$ where 
    $s_{n, \lambda} =  \mathtt{max}\{ \vert\lambda(X)\vert \}_{X \in \Phi_{m, n}}$
    is the largest singular value of $\mathcal{Q}$ over states with $n$ photons,
    $\lambda$ is the eigenfunction associated to $\mathcal{Q}$,
    $n$ is the maximum number of photons in $\psi$, and $\epsilon$ is the estimation precision.
\end{theorem}
\begin{proof}
    Note first that the maximum singular value $s_{n, \lambda}$ exists for any $n \in \mathbb{N}$ and any $\lambda: \mathbb{N}^m \to \mathbb{C}$. 
    Let $\epsilon$ be an estimation precision and $\delta$ a sampling error probability.
    By applying Hoeffding's inequality to the real and imaginary parts of the complex random variable we have:
    $$ P( \vert \overline{E}(\mathcal{Q}) - E(\mathcal{Q}) \vert \geq \epsilon) \leq 4e^{- \frac{T \epsilon^2}{4 s_{n, \lambda}^2}}$$
    Therefore, to achieve the error probability $\delta$ we must sample for time
    $T = \mathtt{log}(\frac{4}{\delta}) \frac{4 s_{n, \lambda}^2}{\epsilon^2}$.
    In other words, we can achieve an estimation precision of $\pm \epsilon s_{n, \lambda}$ in time 
    $T = O(\frac{\mathtt{log}(\frac{1}{\delta})}{\epsilon^2})$.
\end{proof}

A typical problem in our setting appears when the factor of $s_{n, \lambda}$ grows exponentially with $n$.
In this case, we can trade the estimation error $\epsilon$ with sampling time $T$ or error probability $\delta$, 
generally suffering an exponential blow up in each. 
A dual problem arises due to the expectation values becoming exponentially small as we increase the number of photons.
When the expectation value is exponentially small, $0$ provides a good additive approximation.
This is a common problem in quantum systems known as the ``orthogonality catastrophe'' or "sign problem" \cite{anderson_infrared_1967}, 
which we do not study in this paper.

\subsection{Full-counting statistics} \label{sec:full-counting-statistics}

The full-counting-statistics (FCS) offers a series of tools to study the statistics of particle-number distributions,
such as moments and individual probabilities.
In FCS we consider a state $\psi$ subject to a unitary evolution $U$ and we are interested in the characteristic function
\begin{equation}\label{eq-generating-function}
    E(\lambda) =  \sum_{X \in \Phi_{m, n}} e^{i \sum_{j=1}^m \lambda_j X_i} P_U(X \vert \psi) \, ,
\end{equation}
for $\lambda_j \in \mathbb{R}$.
The power of the FCS comes from looking at how the characteristic function varies as a function of $\lambda$.
For simplicity, consider a single parameter $\lambda = \lambda_i$ and set all remaining parameters to zero.
The number operator on mode $i$ is denoted $\hat{n} = \hat{n}_i$.
We use the differentiation operator $\frac{d}{d\lambda} e^{i \lambda \hat{n}} = i \hat{n} e^{i \lambda \hat{n}}$.
Applying this to the characteristic function we obtain the $k$-order moments of the distributions as
$$ \left\langle \hat{n}^k \right\rangle = \bra{\psi} U^\dagger \hat{n}^k U \ket{\psi} = \left[\frac{d^k}{d(i \lambda)^k} E(\lambda)\right]_{\lambda = 0}$$
This corresponds to writing $\chi$ as a generating function
$$E(\lambda) = \sum_{k \in \mathbb{N}} \frac{(i \lambda)^k}{k!} \left\langle \hat{n}^k \right\rangle \, ,$$
which gives us access to important quantities such as mean photon number $\left\langle \hat{n} \right\rangle$, 
variance $\left\langle \hat{n}^2 \right\rangle - \left\langle \hat{n} \right\rangle^2$ and higher-order cumulants
of the distribution.
We can moreover express the photon number distribution as a Fourier transform of the characteristic function.
$$ P_U(X_i = n \, \vert \, \psi) = \frac{1}{2 \pi} \int_{-\pi}^{\pi} E(\lambda) e^{-i \lambda n} d\lambda $$
We see that the FCS gives us access to the relevant statistics of the photon-number distribution as derivatives 
and integrals of $E(\lambda)$ with respect to $\lambda$.

\section{Estimation problems}

We now consider different examples of discrete observables $\mathcal{Q} = \widetilde{U}^\dagger \Lambda \widetilde{U}$
that can be estimated efficiently by sampling from a linear optical circuit.
We compare their performance against classical methods for the same tasks. 
Of particular interest is Gurvits' algorithm \cite{gurvits_complexity_2005,aaronson_generalizing_2012} which provides an estimate of the permanent $\mathtt{Perm}(A)$ of an arbitrary matrix $A$ to within additive error $\pm \epsilon \norm{A}^n$ in $O(\frac{n^2}{\epsilon^2})$ time. 
The bound on the error comes from a bound on the terms in Ryser's formula, which is used as an estimator of the permanent.
This algorithm allows us to approximate expectation values of the form (\ref{eq-expectation-values}) for a large class of observables.

\subsection{Probabilities}\label{sec:probabilities}

The individual probabilities $P_U(X \vert \psi)$ correspond to the expectation values of observables given by $\Lambda = \ket{X} \bra{X}$, $X \in \Phi_{m, n}$.
Since the maximum singular value of $\Lambda$ is $1$, by Proposition~\ref{thm-estimation}, $P_U(X \vert \psi)$ can be estimated in time $O(\frac{1}{\epsilon^2})$ 
with additive error $\pm \epsilon$.
When the initial state is the bosonic product state, we can estimate the probabilities with the same additive error using Gurvits' algorithm.
Indeed, we have that
$$P_U(X \vert 1, \dots, 1) = \vert \bra{X}\widetilde{U}\ket{1, \dots, 1} \vert^2 = \frac{\vert \mathtt{Perm}(U_{X, \overrightarrow{1}}) \vert^2}{\prod_j X_j !}$$
which can be estimated in polynomial time to within additive error $\pm \epsilon \norm{U} = \pm \epsilon$ by Aaronson and Hance's improvement of Gurvits' algorithm \cite{aaronson_generalizing_2012}.

The marginal probabilities $P_U(X_i = I_i \, , \, i \leq k\, \vert \, \psi)$ are the expectation values of observables 
of the form $\Lambda = \mathbb{I}_{m-k} \otimes \ket{I}\bra{I}$ where $I \in \mathbb{N}^k$ is a basis state of $k$ modes.
These observables can be depicted as follows
$$ \scalebox{0.95}
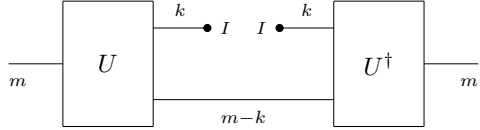
{\begin{tikzpicture}
	\begin{pgfonlayer}{nodelayer}
		\node [style=none] (0) at (-5, -1.75) {};
		\node [style=none] (1) at (-5, 1.75) {};
		\node [style=none] (2) at (-2.5, -1.75) {};
		\node [style=none] (3) at (-2.5, 1.75) {};
		\node [style=none] (4) at (-3.75, 0) {$U$};
		\node [style=none] (7) at (-5, 0) {};
		\node [style=none] (8) at (-6.5, 0) {};
		\node [style=none] (11) at (5, 0) {};
		\node [style=none] (12) at (6.5, 0) {};
		\node [style=none] (13) at (2.5, -1.75) {};
		\node [style=none] (14) at (2.5, 1.75) {};
		\node [style=none] (15) at (5, -1.75) {};
		\node [style=none] (16) at (5, 1.75) {};
		\node [style=none] (17) at (3.75, 0) {$U^\dagger$};
		\node [style=none] (18) at (5, 0) {};
		\node [style=none] (23) at (-6.25, -0.5) {$\scriptstyle m$};
		\node [style=none] (24) at (6.25, -0.5) {$\scriptstyle m$};
		\node [style=none] (25) at (1, 1) {};
		\node [style=none] (26) at (2.5, 1) {};
		\node [style=black node] (27) at (1, 1) {};
		\node [style=none] (28) at (0.5, 1) {$\scriptstyle I$};
		\node [style=none] (29) at (-2.5, 1) {};
		\node [style=none] (30) at (-1, 1) {};
		\node [style=black node] (31) at (-1, 1) {};
		\node [style=none] (32) at (-0.5, 1) {$\scriptstyle I$};
		\node [style=none] (33) at (-2.5, -1) {};
		\node [style=none] (34) at (2.5, -1) {};
		\node [style=none] (35) at (0, -1.5) {$\scriptstyle m - k$};
		\node [style=none] (36) at (1.75, 1.5) {$\scriptstyle k$};
		\node [style=none] (37) at (-1.75, 1.5) {$\scriptstyle k$};
	\end{pgfonlayer}
	\begin{pgfonlayer}{edgelayer}
		\draw (3.center) to (2.center);
		\draw (2.center) to (0.center);
		\draw (0.center) to (1.center);
		\draw (8.center) to (7.center);
		\draw (11.center) to (12.center);
		\draw (16.center) to (15.center);
		\draw (15.center) to (13.center);
		\draw (13.center) to (14.center);
		\draw (1.center) to (3.center);
		\draw (14.center) to (16.center);
		\draw (25.center) to (26.center);
		\draw (29.center) to (30.center);
		\draw (33.center) to (34.center);
	\end{pgfonlayer}
\end{tikzpicture}}$$
When the number of observed modes $k$ is constant, there are polynomial time algorithms to evaluate
the marginal probabilities exactly~\cite[Proposition 68]{aaronson_computational_2011}.
Marginals can also be used to model photon loss in linear optical circuits~\cite{oszmaniec_classical_2018}.


\subsection{Non-interacting observables}\label{sec:non-interacting}

We define non-interacting observables $\mathcal{Q} = \widetilde{Q}$ to be operators generated by an $m \times m$ observable $Q$ on the single-particle space.
We only assume that $Q$ is normal, or equivalently we have $Q= U^\dagger D U$ where $U$ is unitary and $D = \mathtt{diag}(l_1, \dots, l_m)$ is diagonal.
$$ \scalebox{0.95}
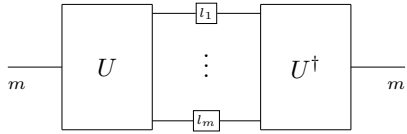
{\begin{tikzpicture}
	\begin{pgfonlayer}{nodelayer}
		\node [style=none] (0) at (-4, -1.75) {};
		\node [style=none] (1) at (-4, 1.75) {};
		\node [style=none] (2) at (-1.5, -1.75) {};
		\node [style=none] (3) at (-1.5, 1.75) {};
		\node [style=none] (4) at (-2.75, 0) {$U$};
		\node [style=none] (5) at (-1.5, 1.5) {};
		\node [style=none] (6) at (-1.5, -1.5) {};
		\node [style=none] (7) at (-4, 0) {};
		\node [style=none] (8) at (-5.5, 0) {};
		\node [style=none] (9) at (1.5, 1.5) {};
		\node [style=none] (10) at (1.5, -1.5) {};
		\node [style=none] (11) at (4, 0) {};
		\node [style=none] (12) at (5.5, 0) {};
		\node [style=none] (13) at (1.5, -1.75) {};
		\node [style=none] (14) at (1.5, 1.75) {};
		\node [style=none] (15) at (4, -1.75) {};
		\node [style=none] (16) at (4, 1.75) {};
		\node [style=none] (17) at (2.75, 0) {$U^\dagger$};
		\node [style=none] (18) at (4, 0) {};
		\node [style=none] (19) at (1.5, 1.5) {};
		\node [style=none] (20) at (1.5, -1.5) {};
		\node [style=box] (21) at (0, 1.5) {$\scriptstyle l_1$};
		\node [style=box] (22) at (0, -1.5) {$\scriptstyle l_m$};
		\node [style=none] (23) at (-5.25, -0.5) {$\scriptstyle m$};
		\node [style=none] (24) at (5.25, -0.5) {$\scriptstyle m$};
		\node [style=none] (25) at (0, 0.25) {$\vdots$};
	\end{pgfonlayer}
	\begin{pgfonlayer}{edgelayer}
		\draw (3.center) to (2.center);
		\draw (2.center) to (0.center);
		\draw (0.center) to (1.center);
		\draw (8.center) to (7.center);
		\draw (11.center) to (12.center);
		\draw (16.center) to (15.center);
		\draw (15.center) to (13.center);
		\draw (13.center) to (14.center);
		\draw (5.center) to (21);
		\draw (6.center) to (22);
		\draw (22) to (20.center);
		\draw (21) to (19.center);
		\draw (1.center) to (3.center);
		\draw (14.center) to (16.center);
	\end{pgfonlayer}
\end{tikzpicture}}$$
The eigenvalues of $\widetilde{Q}$ on the $n$-photon multi-particle space are given by products $\lambda(X) = \prod_i l_i^{X_i}$ where $X \in \Phi_{m, n}$.
Therefore we can bound the singular value as $\absolutevalue{\lambda(X)} \leq \mathtt{max}(\absolutevalue{l_i})^n = \norm{Q}^n$.
By Proposition~\ref{thm-estimation}, the expectation value $E(Q) = \bra{\psi} \widetilde{Q} \ket{\psi}$ for a non-interacting observable $Q$ can
be estimated with additive error $\pm \epsilon \norm{Q}^n$ in time $O(\frac{1}{\epsilon^2})$, by sampling from a linear optical circuit.
When $\psi$ is the bosonic product state, this can also be achieved classically using Gurvits' algorithm introduced above. 
We obtain an additive approximation of $E(Q)$ to $\pm \epsilon \norm{Q}^n$ precision, as this corresponds to evaluating a single permanent
$$\bra{1, \dots, 1}  \widetilde{Q} \ket{1, \dots, 1} = \mathtt{Perm}(Q) \leq \norm{Q}^n \, .$$
More efficient classical approximation algorithms are known when the matrix $Q$ is positive definite or 
positive semi-definite, and when we can bound the variance of the eigenvalues of $Q$, see~\cite{lim_approximating_2023}.

\subsection{Characteristic function}\label{sec:characteristic}

The characteristic function $E(\lambda)$ on state $\widetilde{U} \ket{\psi}$, defined for a vector $\lambda \in \mathbb{R}^m$,
is an example of non-interacting observable that arises in the context of the full-counting-statistics.
The characteristic function defined in \ref{eq-generating-function} corresponds to the expectation of the following diagram.
$$ \scalebox{0.95}
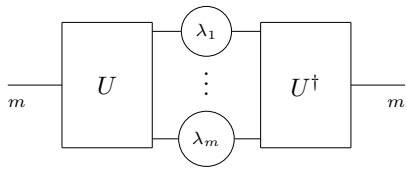
{\begin{tikzpicture}
	\begin{pgfonlayer}{nodelayer}
		\node [style=none] (4) at (-4, -1.75) {};
		\node [style=none] (5) at (-4, 1.75) {};
		\node [style=none] (6) at (-1.5, -1.75) {};
		\node [style=none] (7) at (-1.5, 1.75) {};
		\node [style=none] (8) at (-2.75, 0) {$U$};
		\node [style=none] (12) at (-1.5, 1.5) {};
		\node [style=none] (13) at (-1.5, -1.5) {};
		\node [style=none] (14) at (-4, 0) {};
		\node [style=none] (15) at (-5.5, 0) {};
		\node [style=none] (16) at (1.5, 1.5) {};
		\node [style=none] (17) at (1.5, -1.5) {};
		\node [style=none] (18) at (4, 0) {};
		\node [style=none] (19) at (5.5, 0) {};
		\node [style=none] (24) at (1.5, -1.75) {};
		\node [style=none] (25) at (1.5, 1.75) {};
		\node [style=none] (26) at (4, -1.75) {};
		\node [style=none] (27) at (4, 1.75) {};
		\node [style=none] (28) at (2.75, 0) {$U^\dagger$};
		\node [style=none] (32) at (4, 0) {};
		\node [style=none] (33) at (1.5, 1.5) {};
		\node [style=none] (34) at (1.5, -1.5) {};
		\node [style=new style 0] (35) at (0, 1.5) {$\scriptstyle \lambda_1$};
		\node [style=new style 0] (36) at (0, -1.5) {$\scriptstyle \lambda_m$};
		\node [style=none] (37) at (-5.25, -0.5) {$\scriptstyle m$};
		\node [style=none] (38) at (5.25, -0.5) {$\scriptstyle m$};
		\node [style=none] (39) at (0, 0.25) {$\vdots$};
	\end{pgfonlayer}
	\begin{pgfonlayer}{edgelayer}
		\draw (7.center) to (6.center);
		\draw (6.center) to (4.center);
		\draw (4.center) to (5.center);
		\draw (15.center) to (14.center);
		\draw (18.center) to (19.center);
		\draw (27.center) to (26.center);
		\draw (26.center) to (24.center);
		\draw (24.center) to (25.center);
		\draw (12.center) to (35);
		\draw (13.center) to (36);
		\draw (36) to (34.center);
		\draw (35) to (33.center);
		\draw (5.center) to (7.center);
		\draw (25.center) to (27.center);
	\end{pgfonlayer}
\end{tikzpicture}}$$
It can equivalently be written as
\begin{equation}\label{eq-characteristic-function}
    E(\lambda) = \bra{\psi} \widetilde{U}^\dagger e^{i \sum_{j=0}^m \lambda_j \hat{n}_j}  \widetilde{U} \ket{\psi} = \bra{\psi} \widetilde{U}^\dagger \widetilde{D}(\lambda) \widetilde{U} \ket{\psi}
\end{equation}
where $\lambda_j$ is the j\textsuperscript{th} entry in $\lambda$, $\hat{n}_j$ is the number operator on mode $j$ and $D(\lambda)$ is a diagonal matrix with entries $e^{i \lambda_j}$.
Indeed, one may check that $\widetilde{D}(\lambda) = e^{i \sum_{j=0}^m \lambda_j \hat{n}_j}$.
Since $\norm{D(\lambda)} = 1$, by Proposition~\ref{thm-estimation}, the characteristic function $E(\lambda)$ for a given $\lambda \in \mathbb{R}^m$
can be estimated in time $O(\frac{1}{\epsilon^2})$ with additive error $\pm \epsilon$, by sampling from $\widetilde{U}\ket{\psi}$.
Ivanov and Gurvits~\cite{ivanov_complexity_2020} showed that the characteristic function can be evaluated \emph{exactly} in polynomial time
when the input state is the bosonic product state and the number of measured eigenvalues $\{\lambda_j\}_{j = 1}^k$ is a constant $k \leq m$.
Their result holds for $\lambda_i \in \mathbb{C}$, and therefore for arbitrary non-interacting observables where we only measure a constant number of modes.

\subsection{Polynomials of number operators}\label{sec:polynomials}
The classical simulation and approximation results discussed above are only valid when the input state $\psi$ is a bosonic product state. 
They moreover do not directly allow for the efficient approximation of higher-order classes of discrete observables.
For example, if we want to approximate the first moment $\left\langle \hat{n}_1 \right\rangle =  \bra{\psi} \widetilde{U}^\dagger \Lambda \widetilde{U} \ket{\psi}$, 
this will require estimating the permanent of the matrix defined by the following diagram 
$$ \scalebox{0.95}
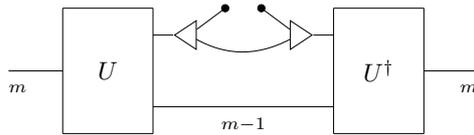
{\begin{tikzpicture}
	\begin{pgfonlayer}{nodelayer}
		\node [style=none] (0) at (-5.25, -2.25) {};
		\node [style=none] (1) at (-5.25, 1.25) {};
		\node [style=none] (2) at (-2.75, -2.25) {};
		\node [style=none] (3) at (-2.75, 1.25) {};
		\node [style=none] (4) at (-4, -0.5) {$U$};
		\node [style=none] (5) at (-5.25, -0.5) {};
		\node [style=none] (6) at (-6.75, -0.5) {};
		\node [style=none] (7) at (4.75, -0.5) {};
		\node [style=none] (8) at (6.25, -0.5) {};
		\node [style=none] (9) at (2.25, -2.25) {};
		\node [style=none] (10) at (2.25, 1.25) {};
		\node [style=none] (11) at (4.75, -2.25) {};
		\node [style=none] (12) at (4.75, 1.25) {};
		\node [style=none] (13) at (3.5, -0.5) {$U^\dagger$};
		\node [style=none] (14) at (4.75, -0.5) {};
		\node [style=none] (15) at (-6.5, -1) {$\scriptstyle m$};
		\node [style=none] (16) at (6, -1) {$\scriptstyle m$};
		\node [style=none] (25) at (-2.75, -1.5) {};
		\node [style=none] (26) at (2.25, -1.5) {};
		\node [style=none] (27) at (-0.25, -2) {$\scriptstyle m - 1$};
		\node [style=none] (28) at (-1.75, 0.5) {};
		\node [style=rmerge] (29) at (1.25, 0.5) {};
		\node [style=black node] (30) at (0.25, 1.25) {};
		\node [style=none] (31) at (2.25, 0.5) {};
		\node [style=lsplit] (32) at (-1.75, 0.5) {};
		\node [style=black node] (33) at (-0.75, 1.25) {};
		\node [style=none] (34) at (-2.75, 0.5) {};
	\end{pgfonlayer}
	\begin{pgfonlayer}{edgelayer}
		\draw (3.center) to (2.center);
		\draw (2.center) to (0.center);
		\draw (0.center) to (1.center);
		\draw (6.center) to (5.center);
		\draw (7.center) to (8.center);
		\draw (12.center) to (11.center);
		\draw (11.center) to (9.center);
		\draw (9.center) to (10.center);
		\draw (1.center) to (3.center);
		\draw (10.center) to (12.center);
		\draw (25.center) to (26.center);
		\draw [bend right] (28.center) to (29);
		\draw (29) to (31.center);
		\draw (30) to (29);
		\draw (32) to (34.center);
		\draw (33) to (32);
	\end{pgfonlayer}
\end{tikzpicture}}$$
Note that this matrix has spectral norm $s_W > 1$ given by the norm of the underlying matrix of the number operator in (\ref{eq-number-op}).
Similarly, the $k$-order moment $\left\langle \hat{n}_1^k \right\rangle$ is given by the permanent of an $(n + k) \times (n + k)$ matrix obtained by conjugating 
by $U$ the matrix $W_k \oplus \mathtt{I}_{m - 1}$ where $W_k$ is the lower-triangular matrix of ones.
The matrices $W_k$ have increasing singular values $s_{W_k} > s_{W_{k-1}}$, and are not even positive semi-definite.  
This results in a classical approximation with error $\pm \epsilon \, s_{W_k}^n$ according to Gurvits algorithm.
In comparison, estimating the moments given access to the quantum distribution $\widetilde{U}\ket{\psi}$ can be achieved in polynomial time 
with additive error $\pm \epsilon$. This is true for any finite degree polynomial built from number operators.

\begin{theorem}\label{thm-number-op}
    The expectation value of a polynomial of number operators $\Lambda = p(\hat{n}_1, \dots, \hat{n}_m)$ 
    of degree $k$, in a state $\widetilde{U}\ket{\psi}$ with at most $n$ photons, 
    can be estimated in time $O(\frac{n^k}{\epsilon^2})$ with additive error $\pm \epsilon$, by sampling from $\widetilde{U}\ket{\psi}$.
\end{theorem}
\begin{proof}
    When $\Lambda = p(\hat{n}_1, \dots, \hat{n}_m)$ the eigenfunction is given by $\lambda(X) = p(X_1, \dots, X_m)$
    which can be bounded as $\absolutevalue{\lambda(X)} \leq C \, n^k$ where $C$ is a constant and $k$ is the degree of $p$ since $X_i \leq n$. 
    Therefore, by Proposition~\ref{thm-estimation}, the expectation value can be estimated in time $O(\frac{1}{\epsilon^2})$ 
    to within additive error $\pm \epsilon \, n^k$. Equivalently, we can reach estimation precision $\epsilon$ in time $O(\frac{n^k}{\epsilon^2})$.
\end{proof}

It is worth emphasizing that polynomials of number operators are in general \emph{unbounded} operators.
Therefore the standard quantum theory --- restricted to Hilbert spaces and bounded linear maps --- does not directly apply in this setting.
For example, we do not know how to recognise operators $\mathcal{Q}$ that can be diagonalised as 
$\mathcal{Q} = \widetilde{U}^\dagger \Lambda \widetilde{U}$ for a single-particle unitary $U$ and a polynomial of number operators $\Lambda$.
We know that $\mathcal{Q}$ must be normal but a stronger condition is needed for such a factorization to be possible.

\section{Differentiation method}\label{sec:diff}

In order to optimise linear optical circuits, we want to be able to evaluate the gradients of different observables with respect to circuit parameters.
For example, one question we want to answer is: what are the values of circuit parameters that maximise the expected number of photons in a mode (or in a subset of the modes)?
To answer these questions we first show how to compute the gradients of the expectation values of non-interacting observables.
Applying this to the characteristic function, we then show that the methods of the FCS carry over to the variational setting and we are able
to (i) evaluate the gradients of $k$-order moments $\frac{d}{d\theta} \left\langle \hat{n}_j^k \right\rangle$
and (ii) estimate the derivative of the individual probabilities $\frac{d}{d\theta} P_U(X \vert \psi)$,
given only access to the derivative of the FCS characteristic function.

\subsection{Gradients of non-interacting observables}\label{sec:circuit-gradient}

Let $Q$ be a normal observable on the single-particle space and
$U\left(\Theta\right)$ be a linear optical circuit parametrised by a list of $N$ phases $\Theta = (\theta_1, \dots, \theta_N)$
with $\theta_i \in [0, 2\pi)$.
We consider the problem of computing the gradient of the expectation value  
\begin{equation}\label{eq-gradient}
    \grad_{\Theta} E(\Theta) = \grad_{\Theta} \bra{\psi}\widetilde{U}^\dagger\left(\Theta\right) \widetilde{Q} \widetilde{U}\left(\Theta\right)\ket{\psi}
\end{equation}
where the $i$th element of $\grad_{\Theta i} = \frac{d}{d\theta_i}$ is the derivative of the expression above with respect to $\theta_i$.
Note that universal interferometers that implement arbitrary unitaries $U(\Theta)$ on $m$ modes only require $N = O(m^2)$ parameters~\cite{reck_experimental_1994,clements_optimal_2016}.
The gradient can therefore be estimated efficiently by dealing with each parameter separately.
We focus on how the expectation value varies as a function of single parameter $\theta$.
By the product rule the derivative has two terms
\begin{align}
    \begin{split}\label{eq-diff-expectation}
        \frac{d}{d\theta} E(\theta)
        = \bra{\psi}\widetilde{U}^\dagger(\theta)\widetilde{Q}\frac{d\widetilde{U}(\theta)}{d\theta}\ket{\psi}\\
        + \bra{\psi}\frac{d \widetilde{U}^\dagger(\theta)}{d\theta} \widetilde{Q} \widetilde{U}(\theta)\ket{\psi} \, .
    \end{split}   
\end{align}
When the observable $Q$ is self-adjoint the second term is just the conjugate of the first, and we recover a real-valued derivative: 
\begin{equation*}
    \frac{d}{d\theta} E(\theta)
    = 2 \, \mathtt{Re}\left(\bra{\psi}\widetilde{U}^\dagger(\theta)\widetilde{Q}\frac{d\widetilde{U}(\theta)}{d\theta}\ket{\psi}  \right)
\end{equation*}
The parameter $\theta$ is a phase shift in the circuit, so we can write 
$U(\theta) = B\, (e^{i\theta} \oplus \mathtt{I}_{m - 1}) \, A$ for unitaries $A$ and $B$.
The differential of a phase shift is given by the equation $\frac{d}{d\theta} e^{i \hat{n} \theta} = i \hat{n} e^{i \hat{n} \theta}$,
which can be expressed diagrammatically as
\begin{equation}\label{eq-diff-phase}
    \scalebox{0.8}{\begin{tikzpicture}
	\begin{pgfonlayer}{nodelayer}
		\node [style=none] (10) at (-5.5, 0) {};
		\node [style=none] (11) at (-1.5, 0) {};
		\node [style=new style 0] (12) at (-3.5, 0) {$\theta$};
		\node [style=none] (18) at (0, 0) {$=$};
		\node [style=new style 0] (22) at (4, 0) {$\theta$};
		\node [style=none] (23) at (2, 0) {};
		\node [style=none] (24) at (1, 0) {$i$};
		\node [style=none] (29) at (4, 0) {};
		\node [style=none] (30) at (10, 0) {};
		\node [style=none] (31) at (6, 0) {};
		\node [style=none] (32) at (9, 0) {};
		\node [style=none] (38) at (6, 0) {};
		\node [style=rmerge] (39) at (9, 0) {};
		\node [style=black node] (40) at (8, 0.75) {};
		\node [style=none] (41) at (8, 1.25) {$\scriptstyle 1$};
		\node [style=lsplit] (42) at (6, 0) {};
		\node [style=black node] (43) at (7, 0.75) {};
		\node [style=none] (44) at (7, 1.25) {$\scriptstyle 1$};
		\node [style=none] (45) at (-7, 0) {\Large $\frac{d}{d\theta}$};
	\end{pgfonlayer}
	\begin{pgfonlayer}{edgelayer}
		\draw (10.center) to (12);
		\draw (12) to (11.center);
		\draw (23.center) to (22);
		\draw (29.center) to (31.center);
		\draw (32.center) to (30.center);
		\draw [bend right] (38.center) to (39);
		\draw [bend left=15] (40) to (39);
		\draw [bend left=15] (42) to (43);
	\end{pgfonlayer}
\end{tikzpicture}} \, .
\end{equation}
By linearity of the differential operator we obtain
$\frac{d}{d \theta} \widetilde{U}(\theta) = i \widetilde{B} ( \hat{n} e^{i\theta \hat{n}} \otimes \mathbb{I}_{m - 1}) \widetilde{A}$.
Then, writing $N = B^\dagger D B$, we can express the first term of (\ref{eq-diff-expectation}) as an expectation value:
\begin{equation}\label{eqn-expectation-diff-diagram}
\scalebox{0.9}{\input{figures/diff-expectation.tikz}}
\end{equation}
where the unitary $U_M$ is obtained via Proposition~\ref{thm-dilation} by dilating the matrix $M=(\mathtt{I}_1\oplus N)\;(W\oplus\mathtt{I}_{m-1})$ 
where $W$ is the underlying matrix of the number operator in Equation~\ref{eq-number-op}.
Since $U_M$ is unitary and thus a normal matrix, it is unitarily diagonalisable $U_M = K_M^\dagger D(u) K_M$ 
for a second unitary $K_M$ and a diagonal matrix $D(u)$, both of size $C = 2m + 2$.
We can then sample from the state $\widetilde{K_M}\ket{\vec{0}_{m + 1}, 1, \psi}$ to estimate the amplitude.
A similar procedure can be repeated to evaluate the second term of (\ref{eq-diff-expectation}) and 
we obtain the following result.

\begin{proposition}\label{thm-gradient}
    Given a state $\psi$ with finitely many photons, and a non-interacting observable $\widetilde{Q}$,
    the gradient of the expectation value $\grad_{\Theta} E(\Theta)$ can be estimated
    by evaluating $O(m^2)$ expectation values from an interferometer with $2m +2$ modes and $n + 1$ photons.
    The setting of parameters for each circuit can be computed in $O(m^3)$ classical time.
    Each expectation value can be estimated within additive error $\pm \epsilon \, (\norm{Q} s_W)^n$ 
    using $O(\frac{1}{\epsilon^2})$ samples of $\widetilde{U}\ket{\psi}$.
\end{proposition}
\begin{proof}
    The algorithm described above and detailed in Appendix \ref{app:diff_alg} can be used to estimate $\frac{d}{d\theta_i} E(\Theta)$ 
    for each parameter $\theta_i$ in the circuit.
    Since a universal inteferometer on $m$ modes has $m (m - 1)$ distinct parameters appearing 
    exactly once in the circuit \cite{reck_experimental_1994},
    the algorithm needs to be repeated $O(m ^2)$ times.
    The classical algorithm for finding circuit parameters 
    requires the diagonalisation of matrices of dimension linear in $m$,
    and the routing of unitary matrices as parameter values, 
    following for example the quadratic algorithm of \cite{clements_optimal_2016}.

    Note that the dilation step requires rescaling the expectation values by a factor of $s_{M}^{n + 1}$, according to Proposition~\ref{thm-dilation}. 
    Using submultiplicativity of the spectral norm, we are only able to bound this quantity as $s_{M} \; \leq \; s_W \,\norm{Q}$,
    where the coefficient $s_W = \frac{1+\sqrt{5}}{2}$ is the spectral norm of the underlying matrix of the number operator given in (\ref{eq-number-op}).
    Following Proposition~\ref{thm-estimation}, we obtain an additive approximation error as above.
\end{proof}

\subsection{Gradients of higher-order observables}\label{sec:observable-gradient}

In Section~\ref{sec:full-counting-statistics} we have seen that the characteristic function gives us access to all the relevant statistics of photon-number distributions as derivatives 
and integrals with respect to eigenvalue parameters $\lambda$.
The parametrization of the expectation value by an eigenvalue $\lambda$ is fundamentally different from the parametrisation
by a circuit parameter $\theta$. Indeed the first only appears once in the expectation, while the second appears twice.
We are interested in how the characteristic function varies as a function of both $\lambda$ and $\theta$:
\begin{equation}
    E(\lambda, \theta) = \bra{\psi} \widetilde{U}^\dagger(\theta) e^{i \sum_{j} \lambda_j \hat{n}_j}\widetilde{U}(\theta) \ket{\psi} \, .
\end{equation}
This corresponds to taking the expectation of the following diagram
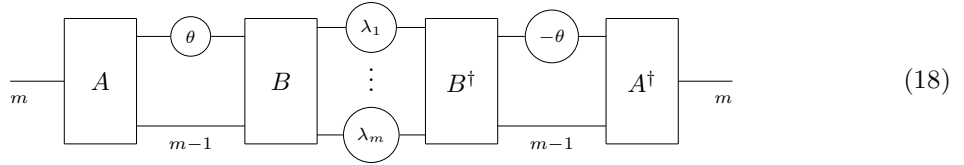
\begin{equation}\label{variational-fcs}
    \scalebox{0.95}{\begin{tikzpicture}
	\begin{pgfonlayer}{nodelayer}
		\node [style=none] (0) at (-3, -1.75) {};
		\node [style=none] (1) at (-3, 1.75) {};
		\node [style=none] (2) at (-1, -1.75) {};
		\node [style=none] (3) at (-1, 1.75) {};
		\node [style=none] (4) at (-2, 0) {$B$};
		\node [style=none] (5) at (-1, 1.5) {};
		\node [style=none] (6) at (-1, -1.5) {};
		\node [style=none] (7) at (-8, 0) {};
		\node [style=none] (8) at (-9.5, 0) {};
		\node [style=none] (9) at (2, 1.5) {};
		\node [style=none] (10) at (2, -1.5) {};
		\node [style=none] (11) at (9, 0) {};
		\node [style=none] (12) at (10.5, 0) {};
		\node [style=none] (13) at (2, -1.75) {};
		\node [style=none] (14) at (2, 1.75) {};
		\node [style=none] (15) at (4, -1.75) {};
		\node [style=none] (16) at (4, 1.75) {};
		\node [style=none] (17) at (3, 0) {$B^\dagger$};
		\node [style=none] (18) at (9, 0) {};
		\node [style=none] (19) at (2, 1.5) {};
		\node [style=none] (20) at (2, -1.5) {};
		\node [style=new style 0] (21) at (0.5, 1.5) {$\scriptstyle \lambda_1$};
		\node [style=new style 0] (22) at (0.5, -1.5) {$\scriptstyle \lambda_m$};
		\node [style=none] (23) at (-9.25, -0.5) {$\scriptstyle m$};
		\node [style=none] (24) at (10.25, -0.5) {$\scriptstyle m$};
		\node [style=none] (25) at (0.5, 0.25) {$\vdots$};
		\node [style=none] (26) at (-8, -1.75) {};
		\node [style=none] (27) at (-8, 1.75) {};
		\node [style=none] (28) at (-6, -1.75) {};
		\node [style=none] (29) at (-6, 1.75) {};
		\node [style=none] (30) at (-7, 0) {$A$};
		\node [style=none] (31) at (-6, 1.25) {};
		\node [style=none] (33) at (-3, 1.25) {};
		\node [style=new style 0] (34) at (-4.5, 1.25) {$\scriptstyle \theta$};
		\node [style=none] (36) at (7, -1.75) {};
		\node [style=none] (37) at (7, 1.75) {};
		\node [style=none] (38) at (9, -1.75) {};
		\node [style=none] (39) at (9, 1.75) {};
		\node [style=none] (40) at (8, 0) {$A^\dagger$};
		\node [style=none] (41) at (4, 1.25) {};
		\node [style=none] (42) at (7, 1.25) {};
		\node [style=new style 0] (43) at (5.5, 1.25) {$\scriptstyle -\theta$};
		\node [style=none] (44) at (-6, -1.25) {};
		\node [style=none] (45) at (-3, -1.25) {};
		\node [style=none] (46) at (4, -1.25) {};
		\node [style=none] (47) at (7, -1.25) {};
		\node [style=none] (48) at (-4.5, -1.75) {$\scriptstyle m - 1$};
		\node [style=none] (49) at (5.5, -1.75) {$\scriptstyle m - 1$};
	\end{pgfonlayer}
	\begin{pgfonlayer}{edgelayer}
		\draw (3.center) to (2.center);
		\draw (2.center) to (0.center);
		\draw (0.center) to (1.center);
		\draw (8.center) to (7.center);
		\draw (11.center) to (12.center);
		\draw (16.center) to (15.center);
		\draw (15.center) to (13.center);
		\draw (13.center) to (14.center);
		\draw (5.center) to (21);
		\draw (6.center) to (22);
		\draw (22) to (20.center);
		\draw (21) to (19.center);
		\draw (1.center) to (3.center);
		\draw (14.center) to (16.center);
		\draw (29.center) to (28.center);
		\draw (28.center) to (26.center);
		\draw (26.center) to (27.center);
		\draw (27.center) to (29.center);
		\draw (31.center) to (34);
		\draw (34) to (33.center);
		\draw (39.center) to (38.center);
		\draw (38.center) to (36.center);
		\draw (36.center) to (37.center);
		\draw (37.center) to (39.center);
		\draw (41.center) to (43);
		\draw (43) to (42.center);
		\draw (44.center) to (45.center);
		\draw (46.center) to (47.center);
	\end{pgfonlayer}
\end{tikzpicture}}
\end{equation}

Moments of the distribution are accessed in FCS by differentiation of $E(\lambda, \theta)$ with respect to $\lambda$. 
Note that $\lambda$ is a cheaper parameter to vary as compared to $\theta$: 
computing $E(\lambda, \theta)$ for different values of $\lambda$ can be done by sampling from the state $\widetilde{U}\psi$ once, 
whereas computing different values of $\theta$ requires preparation of different states $\widetilde{U}(\theta)\psi$.
Since $\lambda$ and $\theta$ are independent, the corresponding differentiation operators commute and we can prove
$$\frac{d}{d\theta} \left\langle \hat{n}_j^k \right\rangle = (-i)^k \left[\frac{d^k}{d\lambda_j^k}  \frac{d}{d\theta} E(\lambda, \theta) \right]_{\lambda = \vec{0}} \, .$$
Following \cite{fan_full-counting_2024}, this derivative can be approximated by finite difference methods using few evaluations of $\frac{d}{d\theta} E(\lambda, \theta)$.
However note that, using the method described in Section~\ref{sec:diff}, evaluating $\frac{d}{d\theta} E(\lambda, \theta)$ for different values of $\lambda$ requires sampling from different states each time. 
This is because the eigenvalues of $Q$ given by $e^{i \lambda_j}$ may not be eigenvalues of the dilated matrix $M$ in Equation~\ref{eq-diff-expectation}.
Alternatively, we can directly apply the differentiation method in Equation~\ref{eq-diff-expectation}
to the diagram in (\ref{variational-fcs}).
This gives a method for evaluating the gradients of $k$-order moments using $k + 1$ ancillary photons and $2 (m + k + 1)$ modes.

Now, we turn our attention to the derivatives of the probabilities of individual events at the output of an interferometer. 
We can extend Equation~\ref{eq-fourier} to the setting with multiple parameters
\begin{equation*}
    P_U(X \vert \psi) = \frac{1}{(2\pi)^m} \int_{-\pi}^\pi E(\lambda, \theta) e^{-i\sum_j\lambda_j X_j} d\lambda_1 \dots d\lambda_m \, .
\end{equation*}
Then, since the bounds of the integral are independent of the value of $\theta$, we can use Leibniz' rule to obtain
\begin{equation}\label{eq-fourier}
    \frac{d}{d\theta} P_U(X \vert \psi) = \frac{1}{(2\pi)^m} \int_{-\pi}^\pi \frac{d E(\lambda, \theta)}{d\theta} e^{-i\sum_j\lambda_j X_j} d\lambda_1 \dots d\lambda_m \, .
\end{equation}
We recover the gradient of the probabilities as a Fourier transform of $\frac{d}{d\theta} E(\lambda, \theta)$. 
In practice, the integral is discretised over a finite number $N$ of equally separated values in the $[-\pi, \pi]$ interval
and approximated by the discrete Fourier transform 
\begin{equation*}
    \frac{d}{d\theta} P_U(X \vert \psi) \approx \frac{1}{(2\pi)^m} \sum_{k _1, \dots, k_m = 0}^N \frac{d E(\frac{2\pi \vec{k}}{N}, \theta)}{d\theta} e^{-i\sum_j \frac{2 \pi k_j}{N} X_j}\, .
\end{equation*}
The complexity of estimating the full $m$-dimensional discrete Fourier transform is $O(N^{2m})$.
As noted in~\cite{ivanov_complexity_2020}, if we are only interested in the marginals for a constant number of modes $k$,
then $N = O(n^k)$ and we can evaluate the marginal probabilities efficiently from the characteristic function,
by a similar reasoning as in Section~\ref{sec:non-interacting}.
In general we may note that the number of required terms $N$ will depend on $\psi$, and in particular on the
number of non-zero terms of Equation~\ref{eq-expectation-values}. Reference ~\cite{fan_full-counting_2024} analyses how large $N$ needs to be for an accurate approximation and proposes filtering
methods for improving efficiency.

\subsection{Example: tunable beam splitter}\label{sec:example}
We follow with an example involving the full-counting-statistics of the tunable beam splitter. 
The expectation $E(\lambda, \theta)$ for the tunable beam splitter with input state $\psi = \ket{0, n}$ is given by the diagram
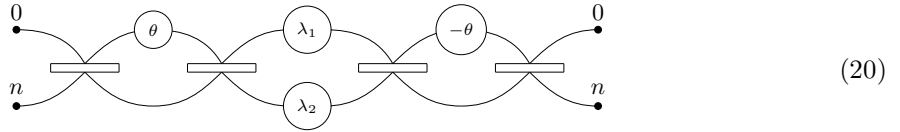
\begin{equation}\label{eqn-expectation-diff-example}
\scalebox{0.9}{\begin{tikzpicture}
	\begin{pgfonlayer}{nodelayer}
		\node [style=black node] (68) at (-6, -1) {};
		\node [style=none] (69) at (-6, 1.25) {};
		\node [style=none] (70) at (-2, 1.25) {};
		\node [style=none] (71) at (-6, -1) {};
		\node [style=none] (72) at (-2, -1) {};
		\node [style=none] (73) at (-5, 0.25) {};
		\node [style=none] (74) at (-3, 0.25) {};
		\node [style=none] (75) at (-5, 0) {};
		\node [style=none] (76) at (-3, 0) {};
		\node [style=none] (77) at (-4, 0.25) {};
		\node [style=none] (78) at (-4, 0) {};
		\node [style=black node] (79) at (-6, 1.25) {};
		\node [style=none] (80) at (-6, 1.75) {$0$};
		\node [style=none] (81) at (-6, -0.5) {$n$};
		\node [style=new style 0] (82) at (-2, 1.25) {$\scriptstyle \theta$};
		\node [style=none] (84) at (-2, 1.25) {};
		\node [style=none] (85) at (2, 1.25) {};
		\node [style=none] (86) at (-2, -1) {};
		\node [style=none] (87) at (2, -1) {};
		\node [style=none] (88) at (-1, 0.25) {};
		\node [style=none] (89) at (1, 0.25) {};
		\node [style=none] (90) at (-1, 0) {};
		\node [style=none] (91) at (1, 0) {};
		\node [style=none] (92) at (0, 0.25) {};
		\node [style=none] (93) at (0, 0) {};
		\node [style=none] (96) at (3, 1.25) {};
		\node [style=none] (97) at (7, 1.25) {};
		\node [style=none] (98) at (3, -1) {};
		\node [style=none] (99) at (7, -1) {};
		\node [style=none] (100) at (4, 0.25) {};
		\node [style=none] (101) at (6, 0.25) {};
		\node [style=none] (102) at (4, 0) {};
		\node [style=none] (103) at (6, 0) {};
		\node [style=none] (104) at (5, 0.25) {};
		\node [style=none] (105) at (5, 0) {};
		\node [style=none] (110) at (7, 1.25) {};
		\node [style=none] (111) at (11, 1.25) {};
		\node [style=none] (112) at (7, -1) {};
		\node [style=none] (113) at (11, -1) {};
		\node [style=none] (114) at (8, 0.25) {};
		\node [style=none] (115) at (10, 0.25) {};
		\node [style=none] (116) at (8, 0) {};
		\node [style=none] (117) at (10, 0) {};
		\node [style=none] (118) at (9, 0.25) {};
		\node [style=none] (119) at (9, 0) {};
		\node [style=black node] (120) at (11, 1.25) {};
		\node [style=black node] (121) at (11, -1) {};
		\node [style=none] (122) at (11, 1.75) {$0$};
		\node [style=none] (123) at (11, -0.5) {$n$};
		\node [style=new style 0] (125) at (7, 1.25) {$\scriptstyle -\theta$};
		\node [style=new style 0] (126) at (2.5, 1.25) {$\scriptstyle \lambda_1$};
		\node [style=new style 0] (127) at (2.5, -1) {$\scriptstyle \lambda_2$};
	\end{pgfonlayer}
	\begin{pgfonlayer}{edgelayer}
		\draw (73.center) to (75.center);
		\draw (75.center) to (76.center);
		\draw (76.center) to (74.center);
		\draw (74.center) to (73.center);
		\draw [in=120, out=0] (69.center) to (77.center);
		\draw [in=-180, out=45] (77.center) to (70.center);
		\draw [in=180, out=-45] (78.center) to (72.center);
		\draw [in=0, out=-135] (78.center) to (71.center);
		\draw (88.center) to (90.center);
		\draw (90.center) to (91.center);
		\draw (91.center) to (89.center);
		\draw (89.center) to (88.center);
		\draw [in=120, out=0] (84.center) to (92.center);
		\draw [in=-180, out=45] (92.center) to (85.center);
		\draw [in=180, out=-45] (93.center) to (87.center);
		\draw [in=0, out=-135] (93.center) to (86.center);
		\draw (100.center) to (102.center);
		\draw (102.center) to (103.center);
		\draw (103.center) to (101.center);
		\draw (101.center) to (100.center);
		\draw [in=120, out=0] (96.center) to (104.center);
		\draw [in=-180, out=45] (104.center) to (97.center);
		\draw [in=180, out=-45] (105.center) to (99.center);
		\draw [in=0, out=-135] (105.center) to (98.center);
		\draw (114.center) to (116.center);
		\draw (116.center) to (117.center);
		\draw (117.center) to (115.center);
		\draw (115.center) to (114.center);
		\draw [in=120, out=0] (110.center) to (118.center);
		\draw [in=-180, out=45] (118.center) to (111.center);
		\draw [in=180, out=-45] (119.center) to (113.center);
		\draw [in=0, out=-135] (119.center) to (112.center);
		\draw (85.center) to (126);
		\draw (126) to (96.center);
		\draw (87.center) to (127);
		\draw (127) to (98.center);
	\end{pgfonlayer}
\end{tikzpicture}}
\end{equation}
By multiplying the matrices in this diagram we obtain a $2 \times 2$ matrix. 
The zero-photon states and effects delete all entries in this matrix except one, 
then this entry is raised to the power $n$ (the number of photons assigned to that entry) and we obtain
$$E(\lambda, \theta) = (e^{i\lambda_1} \mathtt{cos}^2(\frac{\theta}{2}) + e^{i \lambda_2} \mathtt{sin}^2(\frac{\theta}{2}))^n $$
For example, taking $\lambda_1 = 0$, $\lambda_2 = \pi$ we compute the expectation of the $Z$ observable 
$E(Z) = (\mathtt{cos}^2(\frac{\theta}{2}) - \mathtt{sin}^2(\frac{\theta}{2}))^n = \mathtt{cos}^n(\theta)$.
The derivative of $E(Z)$ is $- n \mathtt{sin}(\theta) \mathtt{cos}^{n - 1}(\theta)$.
Using our method, this derivative is evaluated as the following sum of amplitudes
$$\frac{d}{d\theta} E(Z) = \, \scalebox{0.9}{\begin{tikzpicture}
	\begin{pgfonlayer}{nodelayer}
		\node [style=none] (0) at (-4.5, 0.75) {};
		\node [style=none] (1) at (-4.5, 0.75) {};
		\node [style=none] (2) at (-4.5, 0.75) {};
		\node [style=none] (3) at (-4.5, 0.75) {};
		\node [style=none] (4) at (-4.5, -0.75) {};
		\node [style=none] (5) at (-4.5, 1) {};
		\node [style=none] (6) at (-2.5, 1) {};
		\node [style=none] (7) at (-2.5, -1) {};
		\node [style=none] (8) at (-4.5, -1) {};
		\node [style=none] (9) at (-3.5, 0) {$M$};
		\node [style=none] (10) at (-4.5, -0.75) {};
		\node [style=none] (11) at (-2.5, -0.75) {};
		\node [style=none] (12) at (-2.5, 0.75) {};
		\node [style=black node] (13) at (-1.75, 0.75) {};
		\node [style=none] (14) at (-1.25, -0.75) {$\scriptstyle 1$};
		\node [style=black node] (15) at (-5.25, 0.75) {};
		\node [style=none] (16) at (-5.75, -0.75) {$\scriptstyle 1$};
		\node [style=black node] (17) at (-1.75, -0.725) {};
		\node [style=none] (18) at (-1.25, 0.775) {$\scriptstyle n$};
		\node [style=black node] (19) at (-5.25, -0.725) {};
		\node [style=none] (20) at (-5.75, 0.775) {$\scriptstyle n$};
		\node [style=none] (23) at (2.5, 0.75) {};
		\node [style=none] (24) at (2.5, 0.75) {};
		\node [style=none] (25) at (2.5, 0.75) {};
		\node [style=none] (26) at (2.5, 0.75) {};
		\node [style=none] (27) at (2.5, -0.75) {};
		\node [style=none] (28) at (2.5, 1) {};
		\node [style=none] (29) at (4.5, 1) {};
		\node [style=none] (30) at (4.5, -1) {};
		\node [style=none] (31) at (2.5, -1) {};
		\node [style=none] (32) at (3.5, 0) {$M^\dagger$};
		\node [style=none] (33) at (2.5, -0.75) {};
		\node [style=none] (34) at (4.5, -0.75) {};
		\node [style=none] (35) at (4.5, 0.75) {};
		\node [style=black node] (36) at (5.25, 0.75) {};
		\node [style=none] (37) at (5.75, -0.75) {$\scriptstyle 1$};
		\node [style=black node] (38) at (1.75, 0.75) {};
		\node [style=none] (39) at (1.25, -0.75) {$\scriptstyle 1$};
		\node [style=black node] (40) at (5.25, -0.725) {};
		\node [style=none] (41) at (5.75, 0.775) {$\scriptstyle n$};
		\node [style=black node] (42) at (1.75, -0.725) {};
		\node [style=none] (43) at (1.25, 0.775) {$\scriptstyle n$};
		\node [style=none] (44) at (0, 0) {$-$};
		\node [style=none] (45) at (-3.5, 1.5) {$i$};
		\node [style=none] (46) at (3.5, 1.5) {$i$};
	\end{pgfonlayer}
	\begin{pgfonlayer}{edgelayer}
		\draw (5.center) to (8.center);
		\draw (8.center) to (7.center);
		\draw (7.center) to (6.center);
		\draw (6.center) to (5.center);
		\draw (12.center) to (13);
		\draw (11.center) to (17);
		\draw (15) to (3.center);
		\draw (19) to (10.center);
		\draw (28.center) to (31.center);
		\draw (31.center) to (30.center);
		\draw (30.center) to (29.center);
		\draw (29.center) to (28.center);
		\draw (35.center) to (36);
		\draw (34.center) to (40);
		\draw (38) to (26.center);
		\draw (42) to (33.center);
	\end{pgfonlayer}
\end{tikzpicture}}$$
where the matrix $M$ is given by
$$M = \begin{pmatrix} \mathtt{cos}(\theta) & \frac{1}{\sqrt{2}}e^{i \theta} \\ \frac{1}{\sqrt{2}}& 0 \end{pmatrix}$$
These amplitudes are estimated as expectation values of a larger circuit obtained by diagonalising the dilation of $M$ 
given by
\begin{equation}
	U_M = \begin{bmatrix}-\frac{M^\dagger}{s_M} & D_{M^\dagger} \\ D_M & \frac{M}{s_M} \end{bmatrix}
\end{equation}
where $D_M = \sqrt{\mathtt{I}_n - M^\dagger M}$ and $s_M =  \frac{1+\sqrt{5}}{2}$ is the largest singular value of $M$, 
which is independent of $\theta$.

Now we consider the moments of the distribution. The mean number of photons in the first output mode is given by
$$ \left\langle \hat{n}_1\right\rangle = (-i) \left[\frac{d}{d\lambda_1} E(\lambda, \theta)\right]_{\lambda = 0} = n \mathtt{cos}^2(\theta)$$
Using the rules for split and merge given in \cite{de_felice_quantum_2023}, we can obtain this quantity by the diagram manipulation
$$\scalebox{0.9}{\begin{tikzpicture}
	\begin{pgfonlayer}{nodelayer}
		\node [style=none] (0) at (1.75, 1) {};
		\node [style=rmerge] (1) at (4.25, 1) {};
		\node [style=black node] (2) at (3.5, 1.5) {};
		\node [style=none] (3) at (3.5, 2) {$\scriptstyle 1$};
		\node [style=lsplit] (4) at (1.75, 1) {};
		\node [style=black node] (5) at (2.5, 1.5) {};
		\node [style=none] (6) at (2.5, 2) {$\scriptstyle 1$};
		\node [style=rmerge] (8) at (6.5, 0) {};
		\node [style=lsplit] (11) at (-0.5, 0) {};
		\node [style=box] (12) at (0.75, -1) {$\alpha_2$};
		\node [style=box] (13) at (0.75, 1) {$\alpha_1$};
		\node [style=box] (14) at (5.25, -1) {$\alpha_2^*$};
		\node [style=box] (15) at (5.25, 1) {$\alpha_1^*$};
		\node [style=none] (16) at (-1.5, 0) {};
		\node [style=none] (17) at (7.5, 0) {};
		\node [style=none] (18) at (8.5, 0) {$=$};
		\node [style=rmerge] (26) at (16.5, 0) {};
		\node [style=lsplit] (27) at (11.75, 0) {};
		\node [style=box] (28) at (13.25, -1) {$\alpha_2$};
		\node [style=box] (29) at (13.25, 1) {$\alpha_1$};
		\node [style=box] (30) at (15, -1) {$\alpha_2^*$};
		\node [style=none] (31) at (15, 1) {};
		\node [style=none] (32) at (9.5, 0) {};
		\node [style=none] (33) at (18.75, 0) {};
		\node [style=black node] (34) at (18.75, 0) {};
		\node [style=black node] (35) at (9.5, 0) {};
		\node [style=black node] (36) at (-1.5, 0) {};
		\node [style=black node] (37) at (7.5, 0) {};
		\node [style=none] (38) at (-1.5, 0.5) {$\scriptstyle n$};
		\node [style=none] (40) at (7.5, 0.5) {$\scriptstyle n$};
		\node [style=none] (41) at (18.75, 0.5) {$\scriptstyle n$};
		\node [style=none] (42) at (9.5, 0.5) {$\scriptstyle n$};
		\node [style=none] (43) at (7.75, -3) {$n \absolutevalue{\alpha_1}^2$};
		\node [style=none] (44) at (16.25, 1.5) {$ \alpha_1^*$};
		\node [style=rmerge] (46) at (17.75, 0) {};
		\node [style=black node] (47) at (17, 0.75) {};
		\node [style=none] (48) at (17, 1.25) {$\scriptstyle 1$};
		\node [style=lsplit] (49) at (10.5, 0) {};
		\node [style=black node] (50) at (11.25, 0.75) {};
		\node [style=none] (51) at (11.25, 1.25) {$\scriptstyle 1$};
		\node [style=none] (52) at (6, -3) {$=$};
		\node [style=none] (59) at (9.75, -3) {};
		\node [style=none] (60) at (14.75, -3) {};
		\node [style=black node] (61) at (14.75, -3) {};
		\node [style=black node] (62) at (9.75, -3) {};
		\node [style=none] (63) at (14.75, -2.5) {$\scriptstyle n - 1$};
		\node [style=none] (64) at (9.75, -2.5) {$\scriptstyle n - 1$};
		\node [style=box] (67) at (12.25, -3) {$\absolutevalue{\alpha_1}^2 + \absolutevalue{\alpha_2}^2$};
		\node [style=none] (68) at (18, -3) {$n \absolutevalue{\alpha_1}^2$};
		\node [style=none] (69) at (16.25, -3) {$=$};
		\node [style=box] (70) at (15, 1) {$\alpha_1^*$};
		\node [style=none] (71) at (12, 1.5) {$\alpha_1$};
	\end{pgfonlayer}
	\begin{pgfonlayer}{edgelayer}
		\draw [bend right] (0.center) to (1);
		\draw [bend left=15] (2) to (1);
		\draw [bend left=15] (4) to (5);
		\draw [in=-180, out=60, looseness=0.75] (11) to (13);
		\draw (13) to (4);
		\draw (1) to (15);
		\draw [in=120, out=0, looseness=0.75] (15) to (8);
		\draw [in=0, out=-120, looseness=0.75] (8) to (14);
		\draw (14) to (12);
		\draw [in=-60, out=180, looseness=0.75] (12) to (11);
		\draw (16.center) to (11);
		\draw (8) to (17.center);
		\draw [in=-180, out=60, looseness=0.75] (27) to (29);
		\draw [in=120, out=0, looseness=0.75] (31.center) to (26);
		\draw [in=0, out=-120, looseness=0.75] (26) to (30);
		\draw (30) to (28);
		\draw [in=-60, out=180, looseness=0.75] (28) to (27);
		\draw (29) to (31.center);
		\draw [bend left=15] (47) to (46);
		\draw [bend left=15] (49) to (50);
		\draw (35) to (49);
		\draw (49) to (27);
		\draw (26) to (46);
		\draw (46) to (34);
		\draw (62) to (67);
		\draw (67) to (61);
	\end{pgfonlayer}
\end{tikzpicture}}$$
where $\alpha_1 = \mathtt{cos}(\frac{\theta}{2})$ and $\alpha_2 = \mathtt{sin}(\frac{\theta}{2})$.
The derivative of $\left\langle \hat{n}_1\right\rangle$ with respect to $\theta$ is $-\frac{n}{2}\mathtt{sin}(\theta)$. 
Equivalently it is given as
$$\frac{d}{d\theta} \left\langle \hat{n}_1\right\rangle = \, \scalebox{0.9}{\begin{tikzpicture}
	\begin{pgfonlayer}{nodelayer}
		\node [style=none] (0) at (-4.75, 0.75) {};
		\node [style=none] (1) at (-4.75, 0.75) {};
		\node [style=none] (2) at (-4.75, 0.75) {};
		\node [style=none] (3) at (-4.75, 0.75) {};
		\node [style=none] (4) at (-4.75, -0.75) {};
		\node [style=none] (5) at (-4.75, 1) {};
		\node [style=none] (6) at (-2.75, 1) {};
		\node [style=none] (7) at (-2.75, -1) {};
		\node [style=none] (8) at (-4.75, -1) {};
		\node [style=none] (9) at (-3.75, 0) {$M$};
		\node [style=none] (10) at (-4.75, -0.75) {};
		\node [style=none] (11) at (-2.75, -0.75) {};
		\node [style=none] (12) at (-2.75, 0.75) {};
		\node [style=black node] (13) at (-2, 0.75) {};
		\node [style=none] (14) at (-1.25, -0.75) {$\scriptstyle 1$};
		\node [style=black node] (15) at (-5.5, 0.75) {};
		\node [style=none] (16) at (-6, -0.75) {$\scriptstyle 1$};
		\node [style=black node] (17) at (-2, -0.725) {};
		\node [style=none] (18) at (-1.5, 0.775) {$\scriptstyle n$};
		\node [style=black node] (19) at (-5.5, -0.725) {};
		\node [style=none] (20) at (-6, 0.775) {$\scriptstyle n$};
		\node [style=none] (44) at (0, 0) {$-$};
		\node [style=none] (45) at (-3.75, 1.5) {$i$};
		\node [style=none] (47) at (-2.75, 0) {};
		\node [style=none] (48) at (-1.25, 0) {$\scriptstyle 1$};
		\node [style=black node] (49) at (-2, 0.025) {};
		\node [style=none] (53) at (-4.75, 0) {};
		\node [style=none] (54) at (-4.75, 0) {};
		\node [style=none] (55) at (-6, 0) {$\scriptstyle 1$};
		\node [style=black node] (56) at (-5.5, 0.025) {};
		\node [style=none] (57) at (2.5, 0.75) {};
		\node [style=none] (58) at (2.5, 0.75) {};
		\node [style=none] (59) at (2.5, 0.75) {};
		\node [style=none] (60) at (2.5, 0.75) {};
		\node [style=none] (61) at (2.5, -0.75) {};
		\node [style=none] (62) at (2.5, 1) {};
		\node [style=none] (63) at (4.5, 1) {};
		\node [style=none] (64) at (4.5, -1) {};
		\node [style=none] (65) at (2.5, -1) {};
		\node [style=none] (66) at (3.5, 0) {$M^\dagger$};
		\node [style=none] (67) at (2.5, -0.75) {};
		\node [style=none] (68) at (4.5, -0.75) {};
		\node [style=none] (69) at (4.5, 0.75) {};
		\node [style=black node] (70) at (5.25, 0.75) {};
		\node [style=none] (71) at (6, -0.75) {$\scriptstyle 1$};
		\node [style=black node] (72) at (1.75, 0.75) {};
		\node [style=none] (73) at (1.25, -0.75) {$\scriptstyle 1$};
		\node [style=black node] (74) at (5.25, -0.725) {};
		\node [style=none] (75) at (5.75, 0.775) {$\scriptstyle n$};
		\node [style=black node] (76) at (1.75, -0.725) {};
		\node [style=none] (77) at (1.25, 0.775) {$\scriptstyle n$};
		\node [style=none] (78) at (3.5, 1.5) {$i$};
		\node [style=none] (79) at (4.5, 0) {};
		\node [style=none] (80) at (6, 0) {$\scriptstyle 1$};
		\node [style=black node] (81) at (5.25, 0.025) {};
		\node [style=none] (82) at (2.5, 0) {};
		\node [style=none] (83) at (2.5, 0) {};
		\node [style=none] (84) at (1.25, 0) {$\scriptstyle 1$};
		\node [style=black node] (85) at (1.75, 0.025) {};
	\end{pgfonlayer}
	\begin{pgfonlayer}{edgelayer}
		\draw (5.center) to (8.center);
		\draw (8.center) to (7.center);
		\draw (7.center) to (6.center);
		\draw (6.center) to (5.center);
		\draw (12.center) to (13);
		\draw (11.center) to (17);
		\draw (15) to (3.center);
		\draw (19) to (10.center);
		\draw (47.center) to (49);
		\draw (56) to (54.center);
		\draw (62.center) to (65.center);
		\draw (65.center) to (64.center);
		\draw (64.center) to (63.center);
		\draw (63.center) to (62.center);
		\draw (69.center) to (70);
		\draw (68.center) to (74);
		\draw (72) to (60.center);
		\draw (76) to (67.center);
		\draw (79.center) to (81);
		\draw (85) to (83.center);
	\end{pgfonlayer}
\end{tikzpicture}}$$
where
$$M = \begin{pmatrix} 1 & \frac{1}{\sqrt{2}}e^{i \theta} & \frac{i}{2} (e^{i \theta} + 1) \\ \frac{1}{\sqrt{2}} e^{- i \theta} & \frac{i}{\sqrt{2}}& 0 \\ \frac{i}{2} (e^{- i \theta} + 1) & 0 & 0 \end{pmatrix}$$

Finally, we recover the individual probabilities by a Fourier transform
$$P_{\theta}(X_1, X_2) = \int_{-\pi}^\pi \int_{-\pi}^\pi E(\lambda, \theta) e^{-i X_1 \lambda_1} e^{-i X_2 \lambda_2} d\lambda_1 d\lambda_2$$
When $n=2$ we can compute this integral exactly to find $P_{\theta}(2, 0) = \mathtt{cos}^4(\frac{\theta}{2})$,
$P_{\theta}(0, 2) = \mathtt{sin}^4(\frac{\theta}{2})$ and $P_{\theta}(1, 1) = \frac{1}{2} \mathtt{sin}(2\theta)$.
We recover the Hong-Ou-Mandel effect when $\theta = \frac{\pi}{2}$.  
We can moreover compute the derivative of these probabilities using the Fourier transform in Equation~\ref{eq-fourier}. 
For example, $\frac{d}{d\theta} P_{\theta}(1, 1) = \mathtt{cos}(2 \theta)$. 
The sign of this function gives us the direction in which the angle should be rotated to minimise the probability.

\section{Conclusion}\label{sec:conclusion}%

The literature on the complexity of linear optics has largely been focused on sampling problems, as they provide a rigorous advantage \cite{aaronson_computational_2011}.
However, in the contexts of quantum chemistry and machine learning, optimisation landscapes are sometimes easier 
to define in terms of the expectation values of observables \cite{cerezo_variational_2021}, which give rise to estimation (rather than sampling) problems.
Moreover, as mentioned in~\cite{ivanov_complexity_2020}, little evidence has been given for the ``lack of universality'' of bosonic linear optics with arbitrary initial states and (non-adaptive) measurements.
In this work, we developed a theory of observables in linear optics that allows for the comparison of classical and quantum algorithms on a large class of estimation problems.
We then gave different methods to compute the analytic gradients of the expectation values of linear optical circuits with respect to their phase parameters.
The algorithm, in the general setting, relies on the evaluation of the gradient $\grad_\theta E$ of a characteristic function $E(\lambda, \theta)$ of the photon-number distribution $\widetilde{U}(\theta) \ket{\psi}$.
Sampling different values of $\lambda$ then allows to approximate the gradients of higher-order observables of the distribution.
By varying the complexity of the observable $\mathcal{Q}$ and of the state $\psi$, we obtain estimation problems of increasing classical complexity.

Our algorithm for estimating $\grad_\theta E$ requires $O(m^2)$ evaluations of expectation values for circuits of size $2m + 2$ with $n + 1$ photons.
Expectation values are computed via Proposition~\ref{thm-estimation} after performing the 
singular value decomposition of matrices $M$ and $U_M$ whose dimensions are linear in the number of modes $m$.
The finite sampling accuracy that we obtain is $\pm \epsilon (s_W \norm{Q})^n$ where $s_W \approx 1.618$ is a singular value associated to the number operator and $n$ is the number of photons in the circuit.
For the bosonic product state, this matches the estimation error of a classical algorithm that evaluates each partial derivative as a sum of two permanents. 
Therefore, it leaves the problem open of whether more efficient quantum algorithms for estimating these gradients can be achieved.

The natural next step is the implementation of gradient descent on currently available photonic processors 
\cite{maring_general-purpose_2023, pentangelo_high-fidelity_2023}
which could allow us to speed-up machine learning and quantum chemistry experiments.
The results in this paper were tested using the diagrammatic software DisCoPy \cite{toumi_discopy_2022} 
and its interface with Perceval \cite{heurtel_perceval_2023}.
They indicate near-term photonic technology as a promising testing ground for bosonic machine learning landscapes.
Important steps in the direction of physical implementability are analyses of the impact of loss and distinguishability, 
in addition to the finite sampling error studied here.

\section*{Note added}

This is an extended and revised version of a preprint published earlier this year. 
In this short time period, different works have been published on the differentiation of optical circuits \cite{cimini_variational_2024, facelli_exact_2024, pappalardo_photonic_2024, hoch_variational_2024}.
The method presented here is distinct from all the above, as it is based on dilation rather than parameter-shift rules.
As mentioned in \cite{pappalardo_photonic_2024} by the first author, parameter-shift rules are more suitable for near-term implementation thanks to their low sampling complexity.

\section*{Acknowledgements}
The authors would like to thank Tuomas Laakkonen, Pablo Andres-Martinez, Pierre-Emmanuel Emeriau and Richie Yeung for their valuable feedback. 

\bibliographystyle{plain}
\bibliography{preamble/references}

\begin{thebibliography}{10}

\bibitem{aaronson_computational_2011}
Scott Aaronson and Alex Arkhipov.
\newblock The computational complexity of linear optics.
\newblock In {\em Proceedings of the forty-third annual {ACM} symposium on {Theory} of computing}, {STOC} '11, pages 333--342, New York, NY, USA, June 2011. Association for Computing Machinery.

\bibitem{aaronson_generalizing_2012}
Scott Aaronson and Travis Hance.
\newblock Generalizing and {Derandomizing} {Gurvits}'s {Approximation} {Algorithm} for the {Permanent}, December 2012.
\newblock arXiv:1212.0025 [quant-ph].

\bibitem{amaro_filtering_2022}
David Amaro, Carlo Modica, Matthias Rosenkranz, Mattia Fiorentini, Marcello Benedetti, and Michael Lubasch.
\newblock Filtering variational quantum algorithms for combinatorial optimization.
\newblock {\em Quantum Science and Technology}, 7(1):015021, January 2022.

\bibitem{anderson_infrared_1967}
P.~W. Anderson.
\newblock Infrared {Catastrophe} in {Fermi} {Gases} with {Local} {Scattering} {Potentials}.
\newblock {\em Physical Review Letters}, 18(24):1049--1051, June 1967.
\newblock Publisher: American Physical Society.

\bibitem{bartkiewicz_experimental_2020}
Karol Bartkiewicz, Clemens Gneiting, Antonín Černoch, Kateřina Jiráková, Karel Lemr, and Franco Nori.
\newblock Experimental kernel-based quantum machine learning in finite feature space.
\newblock {\em Scientific Reports}, 10(1):12356, July 2020.

\bibitem{benedetti_parameterized_2019}
Marcello Benedetti, Erika Lloyd, Stefan Sack, and Mattia Fiorentini.
\newblock Parameterized quantum circuits as machine learning models.
\newblock {\em Quantum Science and Technology}, 4(4):043001, November 2019.

\bibitem{bradler_certain_2021}
Kamil Bradler and Hugo Wallner.
\newblock Certain properties and applications of shallow bosonic circuits, December 2021.
\newblock arXiv:2112.09766.

\bibitem{bravyi_classical_2021}
Sergey Bravyi, David Gosset, and Ramis Movassagh.
\newblock Classical algorithms for quantum mean values.
\newblock {\em Nature Physics}, 17(3):337--341, March 2021.
\newblock arXiv:1909.11485 [quant-ph].

\bibitem{cerezo_variational_2021}
M.~Cerezo, Andrew Arrasmith, Ryan Babbush, Simon~C. Benjamin, Suguru Endo, Keisuke Fujii, Jarrod~R. McClean, Kosuke Mitarai, Xiao Yuan, Lukasz Cincio, and Patrick~J. Coles.
\newblock Variational quantum algorithms.
\newblock {\em Nature Reviews Physics}, 3(9):625--644, September 2021.

\bibitem{cimini_variational_2024}
Valeria Cimini, Mauro Valeri, Simone Piacentini, Francesco Ceccarelli, Giacomo Corrielli, Roberto Osellame, Nicolò Spagnolo, and Fabio Sciarrino.
\newblock Variational quantum algorithm for experimental photonic multiparameter estimation.
\newblock {\em npj Quantum Information}, 10(1):1--9, February 2024.
\newblock Publisher: Nature Publishing Group.

\bibitem{clements_optimal_2016}
William~R. Clements, Peter~C. Humphreys, Benjamin~J. Metcalf, W.~Steven Kolthammer, and Ian~A. Walmsley.
\newblock Optimal design for universal multiport interferometers.
\newblock {\em Optica}, 3(12):1460--1465, December 2016.

\bibitem{de_felice_quantum_2023}
Giovanni de~Felice and Bob Coecke.
\newblock Quantum {Linear} {Optics} via {String} {Diagrams}.
\newblock {\em Electronic Proceedings in Theoretical Computer Science}, 394:83--100, November 2023.
\newblock arXiv:2204.12985.

\bibitem{del_campo_universal_2018}
Adolfo del Campo.
\newblock Universal {Statistics} of {Topological} {Defects} {Formed} in a {Quantum} {Phase} {Transition}.
\newblock {\em Physical Review Letters}, 121(20):200601, November 2018.
\newblock arXiv:1806.10646 [cond-mat, physics:quant-ph].

\bibitem{facelli_exact_2024}
Giorgio Facelli, David~D. Roberts, Hugo Wallner, Alexander Makarovskiy, Zoë Holmes, and William~R. Clements.
\newblock Exact gradients for linear optics with single photons, September 2024.
\newblock arXiv:2409.16369.

\bibitem{fan_full-counting_2024}
Yun-Zhuo Fan and Dan-Bo Zhang.
\newblock Full-counting statistics of particle distribution on a digital quantum computer.
\newblock {\em Physical Review A}, 109(1):012412, January 2024.
\newblock arXiv:2308.01255 [quant-ph].

\bibitem{feldmann_parallel_2021}
J.~Feldmann, N.~Youngblood, M.~Karpov, H.~Gehring, X.~Li, M.~Stappers, M.~Le~Gallo, X.~Fu, A.~Lukashchuk, A.~S. Raja, J.~Liu, C.~D. Wright, A.~Sebastian, T.~J. Kippenberg, W.~H.~P. Pernice, and H.~Bhaskaran.
\newblock Parallel convolutional processing using an integrated photonic tensor core.
\newblock {\em Nature}, 589(7840):52--58, January 2021.

\bibitem{gan_fock_2022}
Beng~Yee Gan, Daniel Leykam, and Dimitris~G. Angelakis.
\newblock Fock state-enhanced expressivity of quantum machine learning models.
\newblock {\em EPJ Quantum Technology}, 9(1):1--23, December 2022.

\bibitem{grimsley_adaptive_2019}
Harper~R. Grimsley, Sophia~E. Economou, Edwin Barnes, and Nicholas~J. Mayhall.
\newblock An adaptive variational algorithm for exact molecular simulations on a quantum computer.
\newblock {\em Nature Communications}, 10(1):3007, July 2019.

\bibitem{gurvits_complexity_2005}
Leonid Gurvits.
\newblock On the complexity of mixed discriminants and related problems.
\newblock In {\em Proceedings of the 30th international conference on {Mathematical} {Foundations} of {Computer} {Science}}, {MFCS}'05, pages 447--458, Berlin, Heidelberg, August 2005. Springer-Verlag.

\bibitem{halmos1950}
P.R. Halmos.
\newblock Normal dilations and extensions of operators.
\newblock {\em Summa Brasil. Math 2}, 134, 1950.

\bibitem{hamamura_efficient_2020}
Ikko Hamamura and Takashi Imamichi.
\newblock Efficient evaluation of quantum observables using entangled measurements.
\newblock {\em npj Quantum Information}, 6(1):1--8, June 2020.
\newblock Publisher: Nature Publishing Group.

\bibitem{heurtel_perceval_2023}
Nicolas Heurtel, Andreas Fyrillas, Grégoire de~Gliniasty, Raphaël~Le Bihan, Sébastien Malherbe, Marceau Pailhas, Eric Bertasi, Boris Bourdoncle, Pierre-Emmanuel Emeriau, Rawad Mezher, Luka Music, Nadia Belabas, Benoît Valiron, Pascale Senellart, Shane Mansfield, and Jean Senellart.
\newblock Perceval: {A} {Software} {Platform} for {Discrete} {Variable} {Photonic} {Quantum} {Computing}.
\newblock {\em Quantum}, 7:931, February 2023.
\newblock arXiv:2204.00602.

\bibitem{hoch_variational_2024}
Francesco Hoch, Giovanni Rodari, Taira Giordani, Paul Perret, Nicolò Spagnolo, Gonzalo Carvacho, Ciro Pentangelo, Simone Piacentini, Andrea Crespi, Francesco Ceccarelli, Roberto Osellame, and Fabio Sciarrino.
\newblock Variational approach to photonic quantum circuits via the parameter shift rule, October 2024.
\newblock arXiv:2410.06966.

\bibitem{istrati_sequential_2020}
D.~Istrati, Y.~Pilnyak, J.~C. Loredo, C.~Antón, N.~Somaschi, P.~Hilaire, H.~Ollivier, M.~Esmann, L.~Cohen, L.~Vidro, C.~Millet, A.~Lemaître, I.~Sagnes, A.~Harouri, L.~Lanco, P.~Senellart, and H.~S. Eisenberg.
\newblock Sequential generation of linear cluster states from a single photon emitter.
\newblock {\em Nature Communications}, 11(1):5501, October 2020.
\newblock Publisher: Nature Publishing Group.

\bibitem{ivanov_computational_2017}
Dmitri~A. Ivanov.
\newblock Computational complexity of exterior products and multiparticle amplitudes of noninteracting fermions in entangled states.
\newblock {\em Physical Review A}, 96(1):012322, July 2017.
\newblock Publisher: American Physical Society.

\bibitem{ivanov_complexity_2020}
Dmitri~A. Ivanov and Leonid Gurvits.
\newblock Complexity of full counting statistics of free quantum particles in product states.
\newblock {\em Physical Review A}, 101(1):012303, January 2020.
\newblock arXiv:1904.06069 [cond-mat, physics:quant-ph].

\bibitem{knill_scheme_2001}
E.~Knill, R.~Laflamme, and G.~J. Milburn.
\newblock A scheme for efficient quantum computation with linear optics.
\newblock {\em Nature}, 409(6816):46--52, January 2001.

\bibitem{levitov_charge_1993}
L.~S. Levitov and G.~B. Lesovik.
\newblock Charge distribution in quantum shot noise.
\newblock {\em Soviet Journal of Experimental and Theoretical Physics Letters}, 58:230, August 1993.

\bibitem{lim_approximating_2023}
Youngrong Lim and Changhun Oh.
\newblock Approximating outcome probabilities of linear optical circuits.
\newblock {\em npj Quantum Information}, 9(1):1--7, December 2023.
\newblock Publisher: Nature Publishing Group.

\bibitem{maring_general-purpose_2023}
Nicolas Maring, Andreas Fyrillas, Mathias Pont, Edouard Ivanov, Petr Stepanov, Nico Margaria, William Hease, Anton Pishchagin, Thi~Huong Au, Sébastien Boissier, Eric Bertasi, Aurélien Baert, Mario Valdivia, Marie Billard, Ozan Acar, Alexandre Brieussel, Rawad Mezher, Stephen~C. Wein, Alexia Salavrakos, Patrick Sinnott, Dario~A. Fioretto, Pierre-Emmanuel Emeriau, Nadia Belabas, Shane Mansfield, Pascale Senellart, Jean Senellart, and Niccolo Somaschi.
\newblock A general-purpose single-photon-based quantum computing platform, June 2023.
\newblock arXiv:2306.00874.

\bibitem{mezher_solving_2023}
Rawad Mezher, Ana~Filipa Carvalho, and Shane Mansfield.
\newblock Solving graph problems with single photons and linear optics.
\newblock {\em Physical Review A}, 108(3):032405, September 2023.

\bibitem{mitarai_quantum_2018}
Kosuke Mitarai, Makoto Negoro, Masahiro Kitagawa, and Keisuke Fujii.
\newblock Quantum {Circuit} {Learning}.
\newblock {\em Physical Review A}, 98(3):032309, September 2018.

\bibitem{oszmaniec_classical_2018}
Michał Oszmaniec and Daniel~J. Brod.
\newblock Classical simulation of photonic linear optics with lost particles.
\newblock {\em New Journal of Physics}, 20(9):092002, September 2018.
\newblock Publisher: IOP Publishing.

\bibitem{oszmaniec_fermion_2022}
Michał Oszmaniec, Ninnat Dangniam, Mauro~E.S. Morales, and Zoltán Zimborás.
\newblock Fermion {Sampling}: {A} {Robust} {Quantum} {Computational} {Advantage} {Scheme} {Using} {Fermionic} {Linear} {Optics} and {Magic} {Input} {States}.
\newblock {\em PRX Quantum}, 3(2):020328, May 2022.
\newblock Publisher: American Physical Society.

\bibitem{pappalardo_photonic_2024}
Axel Pappalardo, Pierre-Emmanuel Emeriau, Giovanni~de Felice, Brian Ventura, Hugo Jaunin, Richie Yeung, Bob Coecke, and Shane Mansfield.
\newblock A {Photonic} {Parameter}-shift {Rule}: {Enabling} {Gradient} {Computation} for {Photonic} {Quantum} {Computers}, October 2024.
\newblock arXiv:2410.02726.

\bibitem{pentangelo_high-fidelity_2023}
Ciro Pentangelo, Niki Di~Giano, Simone Piacentini, Riccardo Arpe, Francesco Ceccarelli, Andrea Crespi, and Roberto Osellame.
\newblock High-fidelity and polarization insensitive universal photonic processors fabricated by femtosecond laser writing, October 2023.
\newblock arXiv:2310.19718.

\bibitem{peruzzo_variational_2014}
Alberto Peruzzo, Jarrod McClean, Peter Shadbolt, Man-Hong Yung, Xiao-Qi Zhou, Peter~J. Love, Alán Aspuru-Guzik, and Jeremy~L. O'Brien.
\newblock A variational eigenvalue solver on a quantum processor.
\newblock {\em Nature Communications}, 5(1):4213, July 2014.

\bibitem{reck_experimental_1994}
Michael Reck, Anton Zeilinger, Herbert~J. Bernstein, and Philip Bertani.
\newblock Experimental realization of any discrete unitary operator.
\newblock {\em Physical Review Letters}, 73(1):58--61, July 1994.

\bibitem{saggio_experimental_2021}
Valeria Saggio, Beate~E. Asenbeck, Arne Hamann, Teodor Strömberg, Peter Schiansky, Vedran Dunjko, Nicolai Friis, Nicholas~C. Harris, Michael Hochberg, Dirk Englund, Sabine Wölk, Hans~J. Briegel, and Philip Walther.
\newblock Experimental quantum speed-up in reinforcement learning agents.
\newblock {\em Nature}, 591(7849):229--233, March 2021.

\bibitem{schuld_evaluating_2019}
Maria Schuld, Ville Bergholm, Christian Gogolin, Josh Izaac, and Nathan Killoran.
\newblock Evaluating analytic gradients on quantum hardware.
\newblock {\em Physical Review A}, 99(3):032331, March 2019.

\bibitem{shalit_dilation_2020}
Orr Shalit.
\newblock Dilation theory: a guided tour, February 2020.
\newblock arXiv:2002.05596 [math].

\bibitem{shen_deep_2017}
Yichen Shen, Nicholas~C. Harris, Scott Skirlo, Mihika Prabhu, Tom Baehr-Jones, Michael Hochberg, Xin Sun, Shijie Zhao, Hugo Larochelle, Dirk Englund, and Marin Soljačić.
\newblock Deep learning with coherent nanophotonic circuits.
\newblock {\em Nature Photonics}, 11(7):441--446, July 2017.

\bibitem{sparrow_simulating_2018}
Chris Sparrow, Enrique Martín-López, Nicola Maraviglia, Alex Neville, Christopher Harrold, Jacques Carolan, Yogesh~N. Joglekar, Toshikazu Hashimoto, Nobuyuki Matsuda, Jeremy~L. O’Brien, David~P. Tew, and Anthony Laing.
\newblock Simulating the vibrational quantum dynamics of molecules using photonics.
\newblock {\em Nature}, 557(7707):660--667, May 2018.
\newblock Number: 7707 Publisher: Nature Publishing Group.

\bibitem{thomas_efficient_2022}
Philip Thomas, Leonardo Ruscio, Olivier Morin, and Gerhard Rempe.
\newblock Efficient generation of entangled multiphoton graph states from a single atom.
\newblock {\em Nature}, 608(7924):677--681, August 2022.
\newblock Publisher: Nature Publishing Group.

\bibitem{toumi_discopy_2022}
Alexis Toumi, Giovanni de~Felice, and Richie Yeung.
\newblock {DisCoPy} for the quantum computer scientist, May 2022.
\newblock arXiv:2205.05190.

\bibitem{vicary_categorical_2008}
Jamie Vicary.
\newblock A {Categorical} {Framework} for the {Quantum} {Harmonic} {Oscillator}.
\newblock {\em International Journal of Theoretical Physics}, 47(12):3408--3447, December 2008.
\newblock arXiv: 0706.0711.

\bibitem{wein_photon-number_2022}
Stephen~C. Wein, Juan~C. Loredo, Maria Maffei, Paul Hilaire, Abdelmounaim Harouri, Niccolo Somaschi, Aristide Lemaître, Isabelle Sagnes, Loïc Lanco, Olivier Krebs, Alexia Auffèves, Christoph Simon, Pascale Senellart, and Carlos Antón-Solanas.
\newblock Photon-number entanglement generated by sequential excitation of a two-level atom.
\newblock {\em Nature Photonics}, 16(5):374--379, May 2022.
\newblock Publisher: Nature Publishing Group.

\bibitem{wocjan_several_2006}
Pawel Wocjan and Shengyu Zhang.
\newblock Several natural {BQP}-{Complete} problems, June 2006.
\newblock arXiv:quant-ph/0606179.

\end{thebibliography}


\begin{thebibliography}{10}

\bibitem{aaronson_computational_2011}
Scott Aaronson and Alex Arkhipov.
\newblock The computational complexity of linear optics.
\newblock In {\em Proceedings of the forty-third annual {ACM} symposium on {Theory} of computing}, {STOC} '11, pages 333--342, New York, NY, USA, June 2011. Association for Computing Machinery.

\bibitem{amaro_filtering_2022}
David Amaro, Carlo Modica, Matthias Rosenkranz, Mattia Fiorentini, Marcello Benedetti, and Michael Lubasch.
\newblock Filtering variational quantum algorithms for combinatorial optimization.
\newblock {\em Quantum Science and Technology}, 7(1):015021, January 2022.

\bibitem{bartkiewicz_experimental_2020}
Karol Bartkiewicz, Clemens Gneiting, Antonín Černoch, Kateřina Jiráková, Karel Lemr, and Franco Nori.
\newblock Experimental kernel-based quantum machine learning in finite feature space.
\newblock {\em Scientific Reports}, 10(1):12356, July 2020.

\bibitem{bartolucci_fusion-based_2023}
Sara Bartolucci, Patrick Birchall, Hector Bombín, Hugo Cable, Chris Dawson, Mercedes Gimeno-Segovia, Eric Johnston, Konrad Kieling, Naomi Nickerson, Mihir Pant, Fernando Pastawski, Terry Rudolph, and Chris Sparrow.
\newblock Fusion-based quantum computation.
\newblock {\em Nature Communications}, 14(1):912, February 2023.

\bibitem{benedetti_parameterized_2019}
Marcello Benedetti, Erika Lloyd, Stefan Sack, and Mattia Fiorentini.
\newblock Parameterized quantum circuits as machine learning models.
\newblock {\em Quantum Science and Technology}, 4(4):043001, November 2019.

\bibitem{bradler_certain_2021}
Kamil Bradler and Hugo Wallner.
\newblock Certain properties and applications of shallow bosonic circuits, December 2021.
\newblock arXiv:2112.09766.

\bibitem{cerezo_variational_2021}
M.~Cerezo, Andrew Arrasmith, Ryan Babbush, Simon~C. Benjamin, Suguru Endo, Keisuke Fujii, Jarrod~R. McClean, Kosuke Mitarai, Xiao Yuan, Lukasz Cincio, and Patrick~J. Coles.
\newblock Variational quantum algorithms.
\newblock {\em Nature Reviews Physics}, 3(9):625--644, September 2021.

\bibitem{clements_optimal_2016}
William~R. Clements, Peter~C. Humphreys, Benjamin~J. Metcalf, W.~Steven Kolthammer, and Ian~A. Walmsley.
\newblock Optimal design for universal multiport interferometers.
\newblock {\em Optica}, 3(12):1460--1465, December 2016.

\bibitem{de_felice_quantum_2023}
Giovanni de~Felice and Bob Coecke.
\newblock Quantum {Linear} {Optics} via {String} {Diagrams}.
\newblock {\em Electronic Proceedings in Theoretical Computer Science}, 394:83--100, November 2023.
\newblock arXiv:2204.12985.

\bibitem{de_felice_light-matter_2023}
Giovanni de~Felice, Razin~A. Shaikh, Boldizsár Poór, Lia Yeh, Quanlong Wang, and Bob Coecke.
\newblock Light-{Matter} {Interaction} in the {ZXW} {Calculus}.
\newblock {\em Electronic Proceedings in Theoretical Computer Science}, 384:20--46, August 2023.

\bibitem{feldmann_parallel_2021}
J.~Feldmann, N.~Youngblood, M.~Karpov, H.~Gehring, X.~Li, M.~Stappers, M.~Le~Gallo, X.~Fu, A.~Lukashchuk, A.~S. Raja, J.~Liu, C.~D. Wright, A.~Sebastian, T.~J. Kippenberg, W.~H.~P. Pernice, and H.~Bhaskaran.
\newblock Parallel convolutional processing using an integrated photonic tensor core.
\newblock {\em Nature}, 589(7840):52--58, January 2021.

\bibitem{gan_fock_2022}
Beng~Yee Gan, Daniel Leykam, and Dimitris~G. Angelakis.
\newblock Fock state-enhanced expressivity of quantum machine learning models.
\newblock {\em EPJ Quantum Technology}, 9(1):1--23, December 2022.

\bibitem{grimsley_adaptive_2019}
Harper~R. Grimsley, Sophia~E. Economou, Edwin Barnes, and Nicholas~J. Mayhall.
\newblock An adaptive variational algorithm for exact molecular simulations on a quantum computer.
\newblock {\em Nature Communications}, 10(1):3007, July 2019.

\bibitem{halmos1950}
P.R. Halmos.
\newblock Normal dilations and extensions of operators.
\newblock {\em Summa Brasil. Math 2}, 134, 1950.

\bibitem{heurtel_perceval_2023}
Nicolas Heurtel, Andreas Fyrillas, Grégoire de~Gliniasty, Raphaël~Le Bihan, Sébastien Malherbe, Marceau Pailhas, Eric Bertasi, Boris Bourdoncle, Pierre-Emmanuel Emeriau, Rawad Mezher, Luka Music, Nadia Belabas, Benoît Valiron, Pascale Senellart, Shane Mansfield, and Jean Senellart.
\newblock Perceval: {A} {Software} {Platform} for {Discrete} {Variable} {Photonic} {Quantum} {Computing}.
\newblock {\em Quantum}, 7:931, February 2023.
\newblock arXiv:2204.00602.

\bibitem{hughes_training_2018}
Tyler~W. Hughes, Momchil Minkov, Yu~Shi, and Shanhui Fan.
\newblock Training of photonic neural networks through in situ backpropagation.
\newblock {\em Optica}, 5(7):864, July 2018.

\bibitem{knill_scheme_2001}
E.~Knill, R.~Laflamme, and G.~J. Milburn.
\newblock A scheme for efficient quantum computation with linear optics.
\newblock {\em Nature}, 409(6816):46--52, January 2001.

\bibitem{maring_general-purpose_2023}
Nicolas Maring, Andreas Fyrillas, Mathias Pont, Edouard Ivanov, Petr Stepanov, Nico Margaria, William Hease, Anton Pishchagin, Thi~Huong Au, Sébastien Boissier, Eric Bertasi, Aurélien Baert, Mario Valdivia, Marie Billard, Ozan Acar, Alexandre Brieussel, Rawad Mezher, Stephen~C. Wein, Alexia Salavrakos, Patrick Sinnott, Dario~A. Fioretto, Pierre-Emmanuel Emeriau, Nadia Belabas, Shane Mansfield, Pascale Senellart, Jean Senellart, and Niccolo Somaschi.
\newblock A general-purpose single-photon-based quantum computing platform, June 2023.
\newblock arXiv:2306.00874.

\bibitem{mezher_solving_2023}
Rawad Mezher, Ana~Filipa Carvalho, and Shane Mansfield.
\newblock Solving graph problems with single photons and linear optics.
\newblock {\em Physical Review A}, 108(3):032405, September 2023.

\bibitem{mitarai_quantum_2018}
Kosuke Mitarai, Makoto Negoro, Masahiro Kitagawa, and Keisuke Fujii.
\newblock Quantum {Circuit} {Learning}.
\newblock {\em Physical Review A}, 98(3):032309, September 2018.

\bibitem{pai_experimentally_2023}
Sunil Pai, Zhanghao Sun, Tyler~W. Hughes, Taewon Park, Ben Bartlett, Ian A.~D. Williamson, Momchil Minkov, Maziyar Milanizadeh, Nathnael Abebe, Francesco Morichetti, Andrea Melloni, Shanhui Fan, Olav Solgaard, and David A.~B. Miller.
\newblock Experimentally realized in situ backpropagation for deep learning in photonic neural networks.
\newblock {\em Science}, 380(6643):398--404, April 2023.

\bibitem{pentangelo_high-fidelity_2023}
Ciro Pentangelo, Niki Di~Giano, Simone Piacentini, Riccardo Arpe, Francesco Ceccarelli, Andrea Crespi, and Roberto Osellame.
\newblock High-fidelity and polarization insensitive universal photonic processors fabricated by femtosecond laser writing, October 2023.
\newblock arXiv:2310.19718.

\bibitem{peruzzo_variational_2014}
Alberto Peruzzo, Jarrod McClean, Peter Shadbolt, Man-Hong Yung, Xiao-Qi Zhou, Peter~J. Love, Alán Aspuru-Guzik, and Jeremy~L. O'Brien.
\newblock A variational eigenvalue solver on a quantum processor.
\newblock {\em Nature Communications}, 5(1):4213, July 2014.

\bibitem{reck_experimental_1994}
Michael Reck, Anton Zeilinger, Herbert~J. Bernstein, and Philip Bertani.
\newblock Experimental realization of any discrete unitary operator.
\newblock {\em Physical Review Letters}, 73(1):58--61, July 1994.

\bibitem{saggio_experimental_2021}
Valeria Saggio, Beate~E. Asenbeck, Arne Hamann, Teodor Strömberg, Peter Schiansky, Vedran Dunjko, Nicolai Friis, Nicholas~C. Harris, Michael Hochberg, Dirk Englund, Sabine Wölk, Hans~J. Briegel, and Philip Walther.
\newblock Experimental quantum speed-up in reinforcement learning agents.
\newblock {\em Nature}, 591(7849):229--233, March 2021.

\bibitem{schuld_evaluating_2019}
Maria Schuld, Ville Bergholm, Christian Gogolin, Josh Izaac, and Nathan Killoran.
\newblock Evaluating analytic gradients on quantum hardware.
\newblock {\em Physical Review A}, 99(3):032331, March 2019.

\bibitem{shastri_photonics_2021}
Bhavin~J. Shastri, Alexander~N. Tait, T.~Ferreira~de Lima, Wolfram H.~P. Pernice, Harish Bhaskaran, C.~D. Wright, and Paul~R. Prucnal.
\newblock Photonics for artificial intelligence and neuromorphic computing.
\newblock {\em Nature Photonics}, 15(2):102--114, February 2021.

\bibitem{shen_deep_2017}
Yichen Shen, Nicholas~C. Harris, Scott Skirlo, Mihika Prabhu, Tom Baehr-Jones, Michael Hochberg, Xin Sun, Shijie Zhao, Hugo Larochelle, Dirk Englund, and Marin Soljačić.
\newblock Deep learning with coherent nanophotonic circuits.
\newblock {\em Nature Photonics}, 11(7):441--446, July 2017.

\bibitem{slussarenko_photonic_2019}
Sergei Slussarenko and Geoff~J. Pryde.
\newblock Photonic quantum information processing: {A} concise review.
\newblock {\em Applied Physics Reviews}, 6(4):041303, October 2019.

\bibitem{sparrow_simulating_2018}
Chris Sparrow, Enrique Martín-López, Nicola Maraviglia, Alex Neville, Christopher Harrold, Jacques Carolan, Yogesh~N. Joglekar, Toshikazu Hashimoto, Nobuyuki Matsuda, Jeremy~L. O’Brien, David~P. Tew, and Anthony Laing.
\newblock Simulating the vibrational quantum dynamics of molecules using photonics.
\newblock {\em Nature}, 557(7707):660--667, May 2018.
\newblock Number: 7707 Publisher: Nature Publishing Group.

\bibitem{steinbrecher_quantum_2019}
Gregory~R. Steinbrecher, Jonathan~P. Olson, Dirk Englund, and Jacques Carolan.
\newblock Quantum optical neural networks.
\newblock {\em npj Quantum Information}, 5(1):1--9, July 2019.

\bibitem{sz-nagy_harmonic_2010}
Béla Sz.-Nagy, Ciprian Foias, Hari Bercovici, and László Kérchy.
\newblock {\em Harmonic {Analysis} of {Operators} on {Hilbert} {Space}}.
\newblock Springer, New York, NY, 2010.

\bibitem{toumi_discopy_2022}
Alexis Toumi, Giovanni de~Felice, and Richie Yeung.
\newblock {DisCoPy} for the quantum computer scientist, May 2022.
\newblock arXiv:2205.05190.

\bibitem{toumi_diagrammatic_2021}
Alexis Toumi, Richie Yeung, and Giovanni de~Felice.
\newblock Diagrammatic {Differentiation} for {Quantum} {Machine} {Learning}.
\newblock {\em Electronic Proceedings in Theoretical Computer Science}, 343:132--144, September 2021.

\bibitem{zhao_analyzing_2021}
Chen Zhao and Xiao-Shan Gao.
\newblock Analyzing the barren plateau phenomenon in training quantum neural networks with the {ZX}-calculus.
\newblock {\em Quantum}, 5:466, June 2021.
\newblock Publisher: Verein zur Förderung des Open Access Publizierens in den Quantenwissenschaften.

\end{thebibliography}

\appendix

\section{Dilation of the number operator}\label{app:dil_num_op}
To illustrate Proposition \ref{thm-dilation}, we show how the number operator from (\ref{eq-number-op}) can be simulated by a unitary operator.
Note first that the spectral norm of $W = \begin{pmatrix} 0 & 1\\ 1 &1 \end{pmatrix}$ is $ s_W = \norm{W} = \sqrt{\frac{3 + \sqrt{5}}{2}}$ so by Proposition \ref{thm-dilation}, the dilation of the number operator is
\begin{equation}
	U_W = \begin{bmatrix}-\frac{W^\dagger}{s_W} & D_W \\ D_W & \frac{W}{s_W} \end{bmatrix}
\end{equation}
where 
\begin{equation}\nonumber
D_W = \sqrt{3s_W^2-7}\begin{pmatrix} \frac{1}{5}(1+s_W^2) & -\frac{1}{\sqrt{5}}\\ -\frac{1}{\sqrt{5}} & -\frac{1}{5}(s_W^2-4) \end{pmatrix}.
\end{equation}
We can thus express the number operator as the following heralded linear optical circuit:
\begin{align}\label{eq-phase_dilation}
    \begin{split}
        \scalebox{0.8}{\begin{tikzpicture}
	\begin{pgfonlayer}{nodelayer}
		\node [style=none] (30) at (13.25, -1.75) {};
		\node [style=black node] (34) at (20.75, 0.5) {};
		\node [style=none] (35) at (20.75, 1) {$\scriptstyle 0$};
		\node [style=black node] (37) at (14.25, -0.5) {};
		\node [style=none] (39) at (14.25, 0) {$\scriptstyle 1$};
		\node [style=none] (40) at (20, 1.75) {};
		\node [style=none] (41) at (15, 1.75) {};
		\node [style=none] (42) at (20, -2.25) {};
		\node [style=none] (43) at (15, -2.25) {};
		\node [style=none] (44) at (17.5, -0.5) {$\begin{bmatrix}-\frac{W^\dagger}{s_W} & D_W \\ D_W & \frac{W}{s_W} \end{bmatrix}$};
		\node [style=none] (45) at (20.75, 0.5) {};
		\node [style=none] (46) at (14.25, -0.5) {};
		\node [style=none] (47) at (20, 0.5) {};
		\node [style=none] (48) at (15, -0.5) {};
		\node [style=none] (49) at (21.75, -1.75) {};
		\node [style=none] (50) at (20, -1.75) {};
		\node [style=none] (52) at (15, -1.75) {};
		\node [style=black node] (53) at (14.25, 0.5) {};
		\node [style=none] (54) at (14.25, 1) {$\scriptstyle 0$};
		\node [style=none] (55) at (14.25, 0.5) {};
		\node [style=none] (56) at (15, 0.5) {};
		\node [style=black node] (57) at (14.25, 1.5) {};
		\node [style=none] (58) at (14.25, 2) {$\scriptstyle 0$};
		\node [style=none] (59) at (14.25, 1.5) {};
		\node [style=none] (60) at (15, 1.5) {};
		\node [style=black node] (61) at (20.75, 1.5) {};
		\node [style=none] (62) at (20.75, 2) {$\scriptstyle 0$};
		\node [style=none] (63) at (20.75, 1.5) {};
		\node [style=none] (64) at (20, 1.5) {};
		\node [style=black node] (65) at (20.75, -0.5) {};
		\node [style=none] (66) at (20.75, 0) {$\scriptstyle 1$};
		\node [style=none] (67) at (20.75, -0.5) {};
		\node [style=none] (68) at (20, -0.5) {};
		\node [style=none] (71) at (17.5, 1.25) {$s_W$};
	\end{pgfonlayer}
	\begin{pgfonlayer}{edgelayer}
		\draw (40.center) to (42.center);
		\draw (42.center) to (43.center);
		\draw (43.center) to (41.center);
		\draw (41.center) to (40.center);
		\draw (45.center) to (47.center);
		\draw (48.center) to (46.center);
		\draw (49.center) to (50.center);
		\draw (56.center) to (55.center);
		\draw (60.center) to (59.center);
		\draw (63.center) to (64.center);
		\draw (67.center) to (68.center);
		\draw (52.center) to (30.center);
	\end{pgfonlayer}
\end{tikzpicture}}
    \end{split}
\end{align}
The normalisation $s_W$ within the dilation box induces a factor of $s_W^{n+1}$ where $n$ is the number of input photons from the bottom wire and the extra photon comes from the definition of the number operator. 
When we input and post-select $n$ photons with the unbounded number operator from (\ref{eq-phase}) we find
\begin{align}\nonumber
    \begin{split}
        \scalebox{0.8}{\begin{tikzpicture}
	\begin{pgfonlayer}{nodelayer}
		\node [style=none] (21) at (-1.5, 1.25) {};
		\node [style=none] (22) at (-1.5, -1.25) {};
		\node [style=none] (23) at (1.5, -1.25) {};
		\node [style=none] (24) at (1.5, 1.25) {};
		\node [style=none] (25) at (-2.5, -0.75) {};
		\node [style=none] (26) at (2.5, -0.75) {};
		\node [style=none] (27) at (-1.5, -0.75) {};
		\node [style=none] (28) at (1.5, -0.75) {};
		\node [style=none] (29) at (0, 0) {$\begin{pmatrix} 0 & 1\\ 1 &1 \end{pmatrix}$};
		\node [style=black node] (30) at (2.5, 0.75) {};
		\node [style=black node] (32) at (-2.5, 0.75) {};
		\node [style=none] (34) at (-1.5, 0.75) {};
		\node [style=none] (35) at (1.5, 0.75) {};
		\node [style=black node] (60) at (2.5, -0.75) {};
		\node [style=black node] (61) at (-2.5, -0.75) {};
		\node [style=none] (62) at (-2.5, -0.25) {$n$};
		\node [style=none] (63) at (2.5, -0.25) {$n$};
		\node [style=none] (64) at (3.5, 0) {$=$};
		\node [style=none] (65) at (4.5, 0) {$n$};
	\end{pgfonlayer}
	\begin{pgfonlayer}{edgelayer}
		\draw (25.center) to (27.center);
		\draw (21.center) to (22.center);
		\draw (22.center) to (23.center);
		\draw (23.center) to (24.center);
		\draw (24.center) to (21.center);
		\draw (28.center) to (26.center);
		\draw (32) to (34.center);
		\draw (35.center) to (30);
	\end{pgfonlayer}
\end{tikzpicture}}
    \end{split}
\end{align}
and if we input and post-select $n$ photons with the untiary dilation from (\ref{eq-phase_dilation}) we find
\begin{align}\nonumber
 \hspace*{-3mm}
    \begin{split}
        \scalebox{0.8}{\begin{tikzpicture}
	\begin{pgfonlayer}{nodelayer}
		\node [style=none] (30) at (-8.25, -1.5) {};
		\node [style=black node] (34) at (-1.75, 0.75) {};
		\node [style=none] (35) at (-1.75, 1.25) {$\scriptstyle 0$};
		\node [style=black node] (37) at (-8.25, -0.25) {};
		\node [style=none] (39) at (-8.25, 0.25) {$\scriptstyle 1$};
		\node [style=none] (40) at (-2.5, 2) {};
		\node [style=none] (41) at (-7.5, 2) {};
		\node [style=none] (42) at (-2.5, -2) {};
		\node [style=none] (43) at (-7.5, -2) {};
		\node [style=none] (44) at (-5, -0.25) {$\begin{bmatrix}-\frac{W^\dagger}{s_W} & D_W \\ D_W & \frac{W}{s_W} \end{bmatrix}$};
		\node [style=none] (45) at (-1.75, 0.75) {};
		\node [style=none] (46) at (-8.25, -0.25) {};
		\node [style=none] (47) at (-2.5, 0.75) {};
		\node [style=none] (48) at (-7.5, -0.25) {};
		\node [style=none] (49) at (-1.75, -1.5) {};
		\node [style=none] (50) at (-2.5, -1.5) {};
		\node [style=none] (52) at (-7.5, -1.5) {};
		\node [style=black node] (53) at (-8.25, 0.75) {};
		\node [style=none] (54) at (-8.25, 1.25) {$\scriptstyle 0$};
		\node [style=none] (55) at (-8.25, 0.75) {};
		\node [style=none] (56) at (-7.5, 0.75) {};
		\node [style=black node] (57) at (-8.25, 1.75) {};
		\node [style=none] (58) at (-8.25, 2.25) {$\scriptstyle 0$};
		\node [style=none] (59) at (-8.25, 1.75) {};
		\node [style=none] (60) at (-7.5, 1.75) {};
		\node [style=black node] (61) at (-1.75, 1.75) {};
		\node [style=none] (62) at (-1.75, 2.25) {$\scriptstyle 0$};
		\node [style=none] (63) at (-1.75, 1.75) {};
		\node [style=none] (64) at (-2.5, 1.75) {};
		\node [style=black node] (65) at (-1.75, -0.25) {};
		\node [style=none] (66) at (-1.75, 0.25) {$\scriptstyle 1$};
		\node [style=none] (67) at (-1.75, -0.25) {};
		\node [style=none] (68) at (-2.5, -0.25) {};
		\node [style=none] (71) at (-5, 1.5) {$s_W$};
		\node [style=none] (72) at (-8.25, -1.5) {};
		\node [style=black node] (74) at (-1.75, -1.5) {};
		\node [style=black node] (75) at (-8.25, -1.5) {};
		\node [style=none] (76) at (-8.25, -1) {$n$};
		\node [style=none] (77) at (-1.75, -1) {$n$};
		\node [style=none] (96) at (-0.25, 0) {$=$};
		\node [style=none] (193) at (1, 0) {$n$};
	\end{pgfonlayer}
	\begin{pgfonlayer}{edgelayer}
		\draw (40.center) to (42.center);
		\draw (42.center) to (43.center);
		\draw (43.center) to (41.center);
		\draw (41.center) to (40.center);
		\draw (45.center) to (47.center);
		\draw (48.center) to (46.center);
		\draw (49.center) to (50.center);
		\draw (56.center) to (55.center);
		\draw (60.center) to (59.center);
		\draw (63.center) to (64.center);
		\draw (67.center) to (68.center);
		\draw (52.center) to (30.center);
	\end{pgfonlayer}
\end{tikzpicture}}
    \end{split}
\hspace{3mm}
\end{align}
where we have used (\ref{eq-amplitudes}) to solve both diagrams.
While both these operators give the same result it is important to reiterate that the latter case can be implemented on a linear optical circuit, as it is unitary, while the former cannot.

\section{Differentiation algorithm}\label{app:diff_alg}
We give an efficient algorithm for computing the derivative of the expectation value of an $m\times m$ non-interacting observable $Q=U^\dagger D U$
in any state $\psi$ over a finite number $n$ of photons.
To evaluate the first term in Equation~\ref{eq-diff-expectation}, we ffollow the procedure:
\begin{enumerate}[i]
\item compute the matrix $M=(\mathtt{I}_1\oplus B^\dag DB)\;(W\oplus\mathtt{I}_{m-1})$ from the derivation in equation \ref{eqn-expectation-diff-diagram}, this takes time $O(m)$,
\item perform the singular value decomposition of $M$ and construct the Halmos dilation $U_M$ following Proposition~\ref{thm-dilation}, in $O(m^3)$,
\item compute the diagonalisation $U_M = K_M D(u) K_M^\dagger$, in time $O(m^3)$, 
\item route $K_M$ as an interferometer, following the $O(m^2)$ algorithm of \cite{clements_optimal_2016},
\item sample from the circuit $K_M$ on $2m + 2$ modes using the initial state $\ket{\vec{0}_{m+1}, 1, \psi}$ with $1$ ancillary photon to estimate the full probability distribution,
\item compute the expectation value of $D(u)$ from the obtained distribution to find the first term of the derivative of $E(\theta)$.
\end{enumerate}
A similar procedure evaluates the second term in Equation~\ref{eq-diff-expectation} as an expectation value.
We are then able to estimate the derivative by repeating the procedure above twice.

\end{document}